\documentclass[11pt]{article}

\usepackage[margin=1in]{geometry}
\usepackage{graphicx}
\usepackage{caption}
\usepackage{subcaption}
\usepackage{bm}
\usepackage{mathrsfs}
\usepackage[table,xcdraw,dvipsnames]{xcolor}
\usepackage{hyperref}
\hypersetup{
     colorlinks=true,           % false: boxed links; true: colored links
     linkcolor=Blue,          % color of internal links
     citecolor=OliveGreen,            % color of links to bibliography
     filecolor=blue,            % color of file links
     urlcolor=blue,             % color of external links
 }
\usepackage[maxbibnames=99,backend=bibtex,sorting=nyt,style=alphabetic,citestyle=alphabetic, url=false,giveninits=true]{biblatex}
\usepackage{amsmath, amsthm, amssymb, amstext}
\usepackage[nameinlink]{cleveref}
\usepackage{mathtools}
\usepackage{dsfont}
\usepackage[most]{tcolorbox}
\usepackage{changepage}
\newtcolorbox{myframe}[1][]{
  enhanced,
  arc=0pt,
  outer arc=0pt,
  colback=white,
  boxrule=0.8pt,
  #1
}
\usepackage{enumitem}

\usepackage[sfmath]{kpfonts}
\usepackage{mathpazo}
\newcommand{\eps}{\epsilon}
\newcommand{\veps}{\varepsilon}

\newcommand{\norm}[1]{\left\lVert #1 \right\rVert}
\DeclareMathOperator{\Tr}{Tr}
\newcommand{\ket}[1]{| #1 \rangle}
\newcommand{\bra}[1]{\langle #1 |}
\newcommand{\outerprod}[2]{| #1 \rangle \! \langle #2 |}
\newcommand{\id}{\mathrm{id}}
\newcommand{\calH}{\mathcal{H}}
\newcommand{\calE}{\mathcal{E}}
\newcommand{\calN}{\mathcal{N}}
\newcommand{\calM}{\mathcal{M}}
\newcommand{\calU}{\mathcal{U}}
\newcommand{\calW}{\mathcal{W}}

\newcommand{\calV}{\mathcal{V}}

\newcommand{\calB}{\mathcal{B}}
\newcommand{\calT}{\mathcal{T}}
\newcommand{\ee}{\mathrm{e}}
\newcommand{\ii}{\mathrm{i}}

\newcommand{\btheta}{{\bm{\theta}}}
\newcommand{\sfL}{\mathsf{L}}
\newcommand{\sfU}{\mathsf{U}}
\newcommand{\sfD}{\mathsf{D}}

\DeclarePairedDelimiterX{\infdivx}[2]{(}{)}{%
  #1\;\delimsize\|\;#2%
}
\DeclareMathOperator{\Var}{{\textnormal{Var}}}
\DeclareMathOperator{\vect}{{\textnormal{vec}}}

\newcommand{\expect}[1]{\underset{#1}{\expct}}
\DeclareMathOperator{\expct}{{\mathbb E}}

\newtheorem{theorem}{Theorem}[section]
\newtheorem{lemma}[theorem]{Lemma}

\newtheorem{proposition}[theorem]{Proposition}

\newtheorem{claim}[theorem]{Claim}
\theoremstyle{definition}
\newtheorem{definition}[theorem]{Definition}
\newtheorem{protocol}[]{Protocol}
\newtheorem{problem}[]{Problem}

\usepackage{algorithm}
\usepackage{algpseudocode}
\usepackage{tikz}
\usepackage{quantikz}
\usetikzlibrary{shapes,arrows,fit,babel}
\usepackage{microtype}

\bibliography{main}

\newcommand*\samethanks[1][\value{footnote}]{\footnotemark[#1]}

\title{\LARGE{\bf{Optimal quantum circuit cuts with application to clustered Hamiltonian simulation}}}
\author{Aram W. Harrow~\thanks{Center for Theoretical Physics, MIT, Cambridge, MA, 02139, USA}~\footnote{\href{mailto:aram@mit.edu}{aram@mit.edu}}\and Angus Lowe~\samethanks[1]~\footnote{\href{mailto:alowe7@mit.edu}{alowe7@mit.edu}}}
\date{\today}

\begin{document}
\maketitle

\begin{abstract}
    We study methods to replace entangling operations with random local
  operations in a quantum computation, at the cost of increasing the
  number of required executions. First, we consider ``space-like cuts"
  where an entangling unitary is replaced with random local
  unitaries. We propose an entanglement measure for quantum dynamics,
  the \textit{product extent}, which bounds the cost in a procedure
  for this replacement based on two copies of the Hadamard test. In
  the terminology of prior work, this procedure yields a
  quasiprobability decomposition with minimal 1-norm in a number of
  cases, which addresses an open question of Piveteau and Sutter.  As an application, we give an improved algorithm for clustered Hamiltonian simulation.
  Specifically we show that interactions can be removed at a cost which is
  exponential in the sum of their strengths times the evolution time, and vanishing in the limit of weak interactions.

  We also give an improved upper bound on the cost of replacing wires
  with measure-and-prepare channels using ``time-like cuts''.  We prove a matching information-theoretic lower bound when estimating output probabilities.
\end{abstract}

\maketitle

\section{Introduction}

The precise control of a large number of entangled qubits presents a
significant challenge for realizing large-scale quantum computation. While
considerable progress has been made toward the design and construction of
devices which overcome this challenge, near-term quantum computers are likely
to be restricted both in terms of the number of logical qubits available as
well as in their ability to generate and maintain long-range entanglement. In
this work, we study methods which aim to alleviate these issues by replacing
entangling operations with an ensemble of local operations in a given quantum
circuit. Such methods have been referred to collectively as \textit{circuit
    cutting} (e.g., \cite{Lowe2023fast}) or \textit{circuit
    knitting}~\cite{piveteau2023circuit} since, when applied to circuits with an
appropriate structure, they may be employed to simulate large quantum circuits
using circuits defined on strictly fewer qubits and resembling sub-regions of
the original.

Besides the obvious practical motivation for studying these methods, it is also
a long-standing theoretical problem to understand smooth trade-offs between the
classical and quantum resources required to accomplish different information
processing tasks. In quantum Shannon theory, for instance, one often studies
the landscape of achievable rates when trading between generating entanglement,
transmitting classical information, and transmitting quantum information. Prior
work on circuit cutting has suggested that trade-offs between entanglement and
classical randomness also exist in the computational setting. For example,
Ref.~\cite{bravyi2016trading} gives a method for adding ``virtual qubits" in
sparse quantum circuits, while Ref.~\cite{peng2020simulating} develops a
framework for decomposing clustered quantum circuits using mid-circuit Pauli
measurements.

\begin{figure}
    \centering
    \begin{subfigure}[t]{0.3\textwidth}
        \centering
        \begin{tikzpicture}
            \node[] at (0,0) (a1) {};
            \node[above=0.5em of a1] (a) {\includegraphics[width=0.6\textwidth]{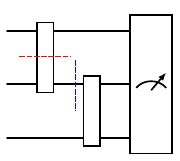}};
        \end{tikzpicture}
        \caption{Original circuit}
    \end{subfigure}
    \hfill
    \begin{subfigure}[t]{0.32\textwidth}
        \centering
        \includegraphics[width=0.8\textwidth]{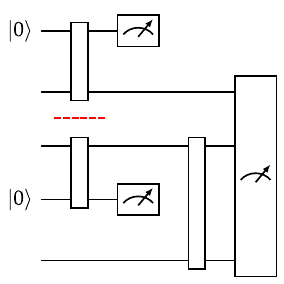}
        \caption{Space-like}
    \end{subfigure}
    \hfill
    \begin{subfigure}[t]{0.33\textwidth}
        \centering
        \begin{tikzpicture}
            \node[] at (0,0) (c1) {};
            \node[above=0.9em of c1] (c) {\includegraphics[width=\textwidth]{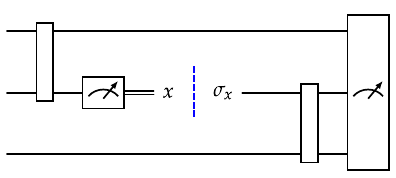}};
        \end{tikzpicture}
        \caption{Time-like}
    \end{subfigure}
    \caption{\label{fig:space_vs_time_fig}We consider two different methods of circuit cutting referred to as space-like and time-like. (a) A quantum circuit with two unitary operations acting on three registers, each comprising some number of qubits. (b) Space-like cut applied to the first unitary operation. (c) Time-like cut applied to the second register, after the first unitary operation.}
\end{figure}

The cost of these and other proposals for circuit cutting is two-fold. First,
the known approaches succeed when the task is to estimate an expectation value,
but it is unclear whether they can be used to sample from the output
distribution of a quantum circuit. This restriction still allows for many
proposed applications of quantum computing, including estimating correlation
functions~\cite{Ortiz_2001}, solving decision problems in $\mathsf{BQP}$, and
optimizing objective functions, e.g., in variational algorithms. The second and
more substantial limitation is an increase in the number of executions of the
computation which tends to grow exponentially in the number of operations
replaced. This exponential cost is to be expected, however, since otherwise one
might envision applying the procedure recursively to derive an efficient, fully
classical algorithm to simulate the circuit. (See~\cite{marshall2023qubit} for
a more detailed heuristic argument along these lines.)

In this paper, we give new methods for two special cases of circuit cutting
which, following Ref.~\cite{Mitarai2021constructing} we refer to as
``space-like" and "time-like", as depicted in \Cref{fig:space_vs_time_fig}. (Despite the terminology, there is no direct connection to relativity.) In
our procedure for space-like cutting, an entangling unitary is replaced by an
ensemble of unitaries acting locally on the original systems and a pair of
ancilla qubits. In a time-like cut, a subset of the wires in a circuit are
replaced by an ensemble of measure-and-prepare operations, i.e., the qubits are
measured and replaced by freshly prepared qubits whose state depends on the
measurement outcome. In both instances we make use of the framework of
quasiprobability decompositions (QPDs), as described in further detail in
\Cref{sec:qpd}.
\subsection{Main results}
\subsubsection{Space-like cuts}
Space-like cuts have been previously studied in, e.g., \cite{yuan2021hybridtn,Mitarai2021constructing,Mitarai2021overhead,piveteau2023circuit}. The cost of cutting a unitary gate can be thought of as a measure of its
entangling power. In \Cref{sec:space_like} we introduce i) a new measure of the
entangling power of unitary operations called the \textit{product extent} and
ii) a simple procedure for space-like cutting whose cost equals the product
extent. In many cases (including all 2-qubit gates, SWAP operators or
transversal operations), we prove that this procedure is optimal. In order to
describe these results quantitatively, it will be helpful to introduce some
preliminary definitions. In the following, $\sfL(\calH_{C})$ denotes a linear operator acting on a quantum system $C$ with Hilbert space $\calH_C$. See \Cref{sec:notation} for a full list of notational conventions.

\begin{definition}[Space-like cut]\label{def:space_like_cut}
    A \emph{space-like cut} of a bipartite quantum channel $\calN_{AB\to AB}$ is a decomposition of the form
    \begin{align}\label{eq:space_like_cut_def}
        \calN & = \sum_{i=1}^m a_i\ (\id_{AB}\otimes \calT_i) \circ \calE_i
    \end{align}
    where
    \begin{itemize}
        %        \item $m$ is a positive integer;
        \item $a\in \mathbb{R}^m$;
        \item
          $\calE_i\colon \sfL(\calH_{AB})\to \sfL(\calH_{AR_ABR_B})$
          %$\calE_i\colon \sfL(\calH_{AB})\to\sfL(\calH_{AR_ABR_B})$
          are quantum channels implementable using local operations and classical communication (LOCC) between $A$ and $B$; and
        \item $\calT_i: \sfL(\calH_{R_AR_B})\to \mathbb{C}$ are \emph{post-processing operations} of the form $\calT_i:X\mapsto \Tr(O X)$ for some Hermitian $O$ such that $O=O^{(A)}_i\otimes O^{(B)}_i$ with $O^{(A)}_i\in\sfL(\calH_{R_A})$ and $O^{(B)}_i\in\sfL(\calH_{R_B})$ and $\norm{O}\leq 1$. These $\calT_i$ are not necessarily quantum operations because they will generally output a density matrix times a scalar.
        \end{itemize}
    We refer to the quantity $\norm{a}_1$ as the \emph{1-norm of the (space-like) cut} and the
    infimum of $\norm{a}_1$ over all space-like cuts as the \emph{gamma factor}
    $\gamma(\calN)$. If, in addition, $\calE_i=\calV_i\otimes\calW_i$ for some isometric channels $\calV_i, \calW_i$ we say that the space-like cut is \emph{local}.
\end{definition}
The gamma factor was previously introduced in Ref.~\cite{piveteau2023circuit}. The form of the decomposition
in a space-like cut is motivated by the fact that such an expression can be leveraged to simulate the action of $\calN$ using the channels appearing in the sum, which do not entangle subsystems $A$ and $B$ (cf.\ \Cref{sec:qpd}). The runtime of this procedure scales with the 1-norm $\norm{a}_1$. With these definitions in hand, we can now state our first result.
\begin{theorem}\label{thm:bound_on_overhead}
    Let $U=\sum_j c_j V_j\otimes W_j$ be a decomposition of $U\in\sfU(\calH_{AB})$ into local unitary operations. The double Hadamard test of \Cref{sec:double_hadamard} is a local space-like cut of $\calU\colon\rho\mapsto U\rho U^\dag$ with two ancilla qubits (i.e., $d_{R_A}=d_{R_B}=2$) and 1-norm $\phi:=2\norm{c}_1^2-\norm{c}_2^2$. Moreover, if this decomposition is an operator Schmidt decomposition\footnote{By an operator Schmidt decomposition, we mean a decomposition of a bipartite operator $X$ acting on $\calH_{AB}$ of the form $X =\sum_j \lambda_j A_j\otimes B_j$ such that $\Tr(A_j^\dag A_k) = d_A \delta_{jk}$, $\Tr(B_j^\dag B_k) = d_B \delta_{jk}$, and $\lambda_j>0$, $\sum_j \lambda_j^2 = 1$. Such a decomposition always exists, though the $A_j, B_j$ need not be unitary.} then
    \begin{align}
        \phi = \gamma(\calU) = 2\norm{c}_1^2 - 1.
    \end{align}
\end{theorem}
This result motivates our definition of
the product extent (\Cref{defn:product_extent}) as the minimum value of
\begin{equation}
    2\norm{c}_1^2-\norm{c}_2^2
    \qquad\text{over all decompositions}\qquad
    U=\sum_j c_j V_j\otimes W_j.
\end{equation}
Our circuit cutting approach is analogous to the stabilizer-rank based classical simulation methods of Ref.~\cite{Bravyi2019simulationofquantum}, where the cost can be related to similar minimizations over decompositions of unitaries.

A potential cause for concern in the motivation for this definition is the
complexity of finding good decompositions, as well as the gate complexity of
implementing the space-like cut using our procedure. Fortunately, the quantity
$\phi$, and hence the product extent, is submultiplicative under composition of
unitaries (again, in a way analogous to \cite{Bravyi2019simulationofquantum}).
This allows one to derive a good decomposition compatible with the double
Hadamard test of \Cref{sec:double_hadamard} using decompositions of the
individual gates which comprise a given circuit.

\begin{figure}
    \centering
    \begin{subfigure}[t]{0.3\textwidth}
        \vskip 0pt
        \centering
        \includegraphics[width=\textwidth]{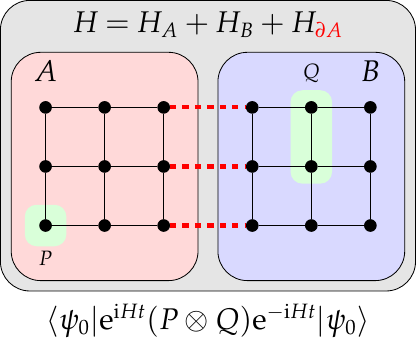}
    \end{subfigure}
    \begin{subfigure}[t]{0.1\textwidth}
        \vskip 0pt
        \vspace{5em}
        \Large{$\ \ \ \ \longrightarrow$}
    \end{subfigure}
    \begin{subfigure}[t]{0.39\textwidth}
        \vskip 0pt
        \centering
        \includegraphics[width=\textwidth]{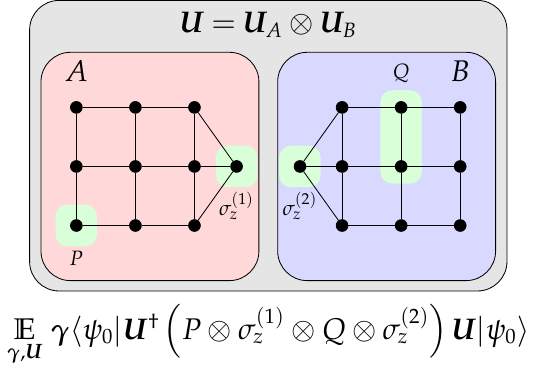}
    \end{subfigure}
    \caption{\label{fig:hamiltonian_intro_fig}The left-hand side depicts the interaction graph of a Hamiltonian $H$ describing a system of qubits with nearest-neighbour
        interactions on a 2D grid. There are 3 weak
        interactions in the boundary $\partial A$, indicated by
        the dashed red lines, and some 3-site (non-local) observable of interest
        $P\otimes Q$ acts on the qubits highlighted in green. The right-hand side shows
        the interaction graph of the pair of randomly-chosen circuits $\bm{U}_A$ and $\bm{U}_B$
        acting locally on $AR_A$ and $BR_B$, respectively, which arises from the procedure described in \Cref{sec:double_hadamard}. By measuring the
        observable $P\otimes Q$ and a product of Pauli-$Z$ operators on the ancillas, we can multiply by a random variable $\bm{\gamma}$ to recover the original mean value, in expectation.}
\end{figure}

We exploit this property in
\Cref{sec:hamiltonian_simulation} to analyze a method for clustered Hamiltonian simulation. This problem was
proposed and analyzed in Ref.~\cite{peng2020simulating} and then further studied in
Ref.~\cite{Childs_2021_trotter_error} as well as Ref.~\cite{sun2022perturbative} under the name ``perturbative quantum simulation". Here, one considers a Hamiltonian $H$
defined on a system of $n$ qubits which is a sum of $\mathrm{poly}(n)$ terms
$H_j$ each of which acts non-trivially on at most $O(1)$ qubits and satisfies
$\norm{H_j}\leq 1$. We further assume that a given term is proportional to a
Pauli operator. Suppose we partition the qubits into two disjoint subsets $A$
and $B$ of $n_A$ and $n_B$ qubits, respectively, and denote the set of
interactions crossing the partition as $\partial A$. In \Cref{thm:clustered_ham_sim} we show that if $\eta:=\sum_{k\in\partial A}\norm{H_k}$ is the interaction strength, then by introducing two ancilla qubits $R_A$ and $R_B$ we can compute the expectation value of a time-evolved observable $\ee^{-\ii H t} (X_A\otimes X_B)\ee^{\ii H t}$ using an ensemble of polynomial-size \emph{local} quantum circuits on $A R_A$ and $BR_B$, with a sample overhead on the order of $\ee^{4\eta t}$, and tending to one as the interaction strength goes to zero, which slightly improves upon the result in~\cite{sun2022perturbative}, to be discussed shortly. The procedure in the theorem is depicted
schematically in \Cref{fig:hamiltonian_intro_fig}.

For clustered Hamiltonians whose interaction strength between the partitions
is sufficiently small, and for short enough times, we envision that the sample overhead may be
manageable in practical settings. It is also clear that we may execute these local
circuits one-at-a-time so long as the initial state is
a product state $\rho=\rho_A\otimes \rho_B$. This results in an algorithm which requires executing quantum circuits defined on just $\max\{n_A,n_B\}$ qubits and a single ancilla qubit.

Our result gives a rigorous proof that a runtime on the order of $\ee^{4\eta t}$ is possible for simulating
bipartite Hamiltonians on qubit-limited devices, up to polynomial factors. This runtime was previously derived using Hadamard tests and different decomposition methods in \cite{sun2022perturbative}, dramatically improving on the earlier works~\cite{peng2020simulating, Childs_2021_trotter_error}. However, the estimator constructed in~\cite{sun2022perturbative} has a small additional bias as well as a multiplicative overhead which does not tend to one in the limit of weak interactions\footnote{This comes from a different definition of cost in \cite{sun2022perturbative}, where a unit cost is assigned to terms which may require multiple Hadamard tests to estimate.}. Additionally, their decomposition technique works well for weakly interacting Hamiltonians, but gives too high a cost for more generic unitaries. For instance, the optimal cost (1-norm) in simulating a CNOT gate is known to be $\gamma(\textnormal{CNOT})=3$ and is recovered using our method, whereas the cost obtained using the decomposition in~\cite{sun2022perturbative} is approximately $9.6$\footnote{In the language of this paper, the cost stated in Eq.~(53) in the Supplementary Material of \cite{sun2022perturbative} is $2\ee^{2\theta}$ for any Pauli rotation $\ee^{\ii \theta P}$ for some Pauli $P$. Since, up to local unitaries, a CNOT can be written as such a rotation with $\theta=\pi/4$, we get a cost of $2\ee^{\pi/2}\approx 9.6$.} . Our contribution here is to remedy these issues through a unification of the ideas in~\cite{sun2022perturbative} and prior work on circuit cutting, e.g., \cite{piveteau2023circuit}.

Earlier, an
upper bound of $2^{O(\eta^2 t^2 |\partial A|/\eps)}$ was proven in
Ref.~\cite{peng2020simulating}. There $\eps$ is the precision
attained by the Trotter formula, and the $O(\cdot)$ in the exponent
hides a very large constant as well as a
dependence on the degree of the interaction graph. By invoking
higher-order Trotter formulae, Ref.~\cite{Childs_2021_trotter_error}
improved this to $2^{O(\eta^{1/p}t^{1+1/p}|\partial A|/\eps^{1/p})}$,
with the same constant in the exponent and where $p$ indicates the
order of the product formula. Furthermore, in the latter two previous
works it was assumed that the terms in the Hamiltonian act
geometrically locally, which is an assumption we are able to drop in
our scheme. Noting that $\eta\leq |\partial A|$, both bounds are
strictly worse than $\ee^{4\eta t}$. However, these bounds are valid even in the case of arbitrarily many subsystems (up to polynomial factors), whereas we focus on bipartite Hamiltonians. Although \cite{sun2022perturbative} suggests a procedure for $N$ subsystems, it is unclear at the time of writing whether the overhead from such a procedure would grow exponentially in $N$. We leave the task of providing more precise bounds on the cost of decomposing circuits into more than two subsytems to future work.

\subsubsection{Time-like cuts}
In \Cref{sec:time_like} we give an improved bound on the cost of replacing
wires in a circuit with measure-and-prepare operations, and we show that this
bound is tight in some cases. Such time-like cuts were previously studied in~\cite{peng2020simulating,Lowe2023fast} where they were applied to quantum optimization and Hamiltonian simulation algorithms, and were further analyzed in~\cite{brenner2023optimal,Mitarai2021constructing}. For Hamiltonian simulation, time-like cuts may be preferable in cases where the interaction graph has a small vertex separator, but no small edge separator.

The upper bound makes use of a certain time-like
cut.
\begin{definition}[Time-like cut]\label{def:time_like_cut}
    A \emph{time-like cut} of a quantum channel $\calN_{A\to B}$ is a decomposition of the form
    \begin{align}\label{eq:time_like_cut_defn}
        \calN & = \sum_{i=1}^m a_i\ (\id_{B}\otimes \calT_i) \circ \calM_i
    \end{align}
    where
    \begin{itemize}
        %        \item $m$ is a positive integer;
        \item $a\in \mathbb{R}^m$;
        \item $\calM_i\colon \sfL(\calH_{A})\to \sfL(\calH_{BR_B})$ are measure-and-prepare channels; and
        \item $\calT_i: \sfL(\calH_{R_B})\to \mathbb{C}$ are post-processing operations of the form $\calT_i:X\mapsto \Tr(O X)$ for some Hermitian $O$ such that $\norm{O}\leq 1$.
    \end{itemize}
    We refer to the quantity $\norm{a}_1$ as the \emph{1-norm of the (time-like) cut} and the
    infimum of $\norm{a}_1$ over all time-like cuts as the \emph{time-like gamma factor}
    $\gamma_{\uparrow}(\calN)$.
\end{definition}
We remark that there is a way to unify the terminology in the time- and space-like cases using
the formalism presented in, for example, Ref.~\cite{Gour2021entanglementbipartite}. However,
in this work we find it more convenient to treat the two cases separately.

In \Cref{thm:rank_dep_wire_cutting} we show that, to estimate an observable
$X$ with respect to the output state of a generic circuit, one can use a
decomposition of the form in \Cref{def:time_like_cut} to replace $k$ wires in
the circuit with an ensemble of measure-and-prepare channels while increasing
i) the size of the circuit by at most $O(k^2)$ additional gates and ii) the
number of executions by a multiplicative factor of at most $O(2^k r)$ if $X$
has rank at most $r\leq 2^k$, and by at most $O(4^k)$ in general. In
particular, a cost of $O(2^k)$ executions suffices for additive-error estimates
of the output probabilities of the original circuit. Interestingly, the
proposal involves only random \textit{diagonal} unitary $2$-designs as well as
state preparation and measurement in the computational basis.

We then give an information-theoretic argument that any similar procedure
necessarily increases the number of required samples by a factor of at least
$\Omega(2^k)$, leading to \Cref{thm:time_like_lower_bound}. This argument is
somewhat reminiscent of quantum data hiding~\cite{divincenzo2002datahiding},
wherein a bit can be perfectly encoded in a random choice of mixed state shared
between Alice and Bob, but remains inaccessible so long as they use
measurements implementable with local operations and classical communication
(LOCC). Crucially, the ``data hiding states" can be chosen to be pure states in
our case since we consider a heavily restricted class of LOCC measurements,
whereas it is known that the states must be mixed in order to hide the bit
against LOCC more generally~\cite{Matthews2009distinguishability}. We leverage
this difference to prove the lower bound when estimating output probabilities.

\subsection{Related work}
Ref.~\cite{piveteau2023circuit} gives a procedure which achieves the minimal
1-norm in a QPD for the special case of bipartite Clifford unitaries. The
procedure we give in \Cref{sec:space_like} differs in a few key respects: i)
the upper bound on the 1-norm of our procedure is applicable to arbitrary
unitaries, ii) our procedure makes use of a single pair of ancilla qubits,
rather than a number of ancilla qubits growing with the dimension of the
unitary, iii) our procedure does not use classical communication between
parties, and iv) the overhead in gate complexity of our procedure is explicitly
shown to be small in relevant cases. The procedure in \Cref{sec:space_like} is
also closely related to ideas in Refs.~\cite{bravyi2016trading}
and~\cite{eddins2022doubling}. In Ref.~\cite{bravyi2016trading}, a circuit
resembling a Hadamard test is used to simulate $k$ physical qubits in sparse
quantum circuits defined on $n+k$ qubits. Roughly speaking, this would
correspond to classically simulating a $k$-qubit subsystem of the ``double
Hadamard test" circuit depicted in \Cref{fig:double_hadamard_test}, and would
therefore not be applicable in the settings we consider, where $k$ may be equal
to $n$ in the worst case. In Ref.~\cite{eddins2022doubling}, a QPD-based method
is suggested for ``doubling" the size of a quantum simulation, though their
analysis is performed at the level of quantum states rather than unitaries.
Another key technical difference is that our procedure does not require
preparing states corresponding to those in a Schmidt decomposition of the state produced by the circuit, which underlies the application of our result to
clustered Hamiltonian simulation.

The decomposition of the identity channel into measure-and-prepare channels
which we employ in \Cref{sec:good_mp_channels} has previously been used to
obtain similar results. To the best of our knowledge, the decomposition was
first explicitly given in~\cite{Yuan2021universalmemories} to describe an
application of the dynamic entanglement measure which they introduce. The
authors show that the resulting QPD can be used to estimate expectation values
of observables, though they do not provide explicit implementations of the
channels. A similar procedure which makes use of ancilla qubits was then
suggested in~\cite{brenner2023optimal}. Follow-up
works~\cite{harada2023optimal,pednault2023alternative} removed the need for
ancilla qubits and provided explicit gate complexities for implementing the
relevant channels. The procedure we give in \Cref{sec:good_mp_channels} makes
use of the same decomposition as in the works above, though we give a simpler
implementation of the channels based on diagonal 2-designs (cf.\
\cite[Algorithm 1]{harada2023optimal} versus \Cref{prot:prot_1}). All in all,
our improved upper bound comes from an improved \textit{analysis} of the
procedure, rather than a different choice of measure-and-prepare channels. Our
lower bound in this setting is not implied by lower bounds on the 1-norm
appearing in prior work~\cite{brenner2023optimal}, as discussed in
\Cref{sec:qpd}.

During the preparation of this paper we became aware of two other works whose
results overlap with some aspects of our procedure for space-like cutting.
In~\cite[{Theorem~5.1}]{schmitt2023cutting} the authors show using a different
analysis that for bipartite unitaries (referred to as ``KAK-like" unitaries in
their work), the minimal 1-norm in a QPD is at most the product
extent\footnote{Our description of their result is rewritten in the language
    of this paper.} defined in
\Cref{sec:space_like}, which overlaps with the content of
\Cref{thm:bound_on_overhead}. %\Cref{thm:bound_on_overhead} goes beyond this
%work by giving explicit circuits and bounds on the number of ancilla qubits
%required.
Ref.~\cite{ufrecht2023optimal} gives a similar set of results for the
special case of 2-qubit rotation gates.

%--- Preliminaries ---%
\section{Preliminaries}
\subsection{Notation}\label{sec:notation}
\paragraph{Sets.}
Throughout, we let $\calH_A, \calH_B$, etc. denote finite-dimensional Hilbert
spaces representing quantum systems $A, B$, etc., and we denote their
dimensions by $d_A$, $d_B$, etc. respectively. We denote by
$\mathsf{L}(\calH_A,\calH_B)$ the set of all linear operators from $\calH_A$ to
$\calH_B$, and $\mathsf{L}(\calH_A)$ the set of all square linear operators
acting on $\calH_A$. We let $\mathsf{D}(\calH_A)\subset \mathsf{L}(\calH_A)$ be
the set of all quantum states of system $A$, and $\mathsf{U}(\calH_A)\subset
    \mathsf{L}(\calH_A)$ be the set of unitary operators acting on $\calH_A$. The
set of separable states (i.e., convex combinations of tensor product states) on
the bipartite Hilbert space $\calH_{AB}=\calH_A\otimes \calH_B$ will be denoted
by $\mathsf{SEP}(\calH_{AB}|A,B)$.

\paragraph{Vectorization, Choi and Bell states.}
The notation $\vect(\cdot)$ denotes vectorization, i.e., the natural linear
bijection from $\sfL(\calH_A,\calH_B)$ to
$\calH_B\otimes\calH_A$. For a given quantum channel $\calN:
    \mathsf{L}(\calH_A)\to \mathsf{L}(\calH_{A^\prime})$, we let $J_\calN \in
    \mathsf{L}(\calH_A\otimes \calH_{A^\prime})$ be the Choi-Jamiolkowski state
(referred to as the \textit{Choi state} from now on) corresponding to the
channel, i.e., $J_\calN = (\mathrm{id}_A\otimes \calN)(\Phi_A)$ where $\Phi_A =
    \frac{1}{d_A} \vect(\mathds{1}_{A})\vect(\mathds{1}_{A})^\dag$. For two quantum
systems of equal dimension $A$ and $A^\prime$ we let $\ket{\Phi}_{AA^\prime}:=
    \frac{1}{\sqrt{d_A}}\sum_{j\in [d_A]}\ket{j}_A\otimes \ket{j}_{A^\prime}$ be
the maximally entangled state (or Bell state) between $A$ and $A^\prime$.

\paragraph{Some special operations.}
The SWAP operator acting on the bipartite Hilbert space $(\mathbb{C}^{d})^{\otimes 2}$ is denoted by $F$ and has the action $F \ket{\psi}\otimes
    \ket{\phi} = \ket{\phi}\otimes \ket{\psi}$ for any $\ket{\psi},\ket{\phi}\in
    \mathbb{C}^d$. For $n$-partite Hilbert spaces we denote by $F_{\pi}$ the
permutation operator corresponding to the permutation $\pi\in \mathfrak{S}_n$
where $\mathfrak{S}_n$ is the symmetric group of order $n$. The partial
transpose map applied to an operator $X\in\sfL(\calH_{AB})$ is denoted
$X^\Gamma$ and has the action $\outerprod{i}{j}_{A}\otimes
    \outerprod{k}{\ell}_{B}\mapsto \outerprod{i}{j}_{A}\otimes
    \outerprod{\ell}{k}_{B}$ in the standard basis. For an operator
$M\in\mathsf{L}(\calH_A)$ let $\overline{M}$ denote the complex conjugate of
$M$. If an operator $M$ acts on a tensor product Hilbert space then $\Tr_j(M)$ denotes the partial trace over the $j^\text{th}$ subsystem in the tensor product. For example, if $M=A\otimes B\otimes C$ then $\Tr_2(M) = \Tr(B) A\otimes C$.

\paragraph{Random variables, distributions.}
We denote random variables, including matrix-valued random variables, using
bold font e.g., $\bm{x}$, $\bm{U}$, etc. If $\bm{x}$ is a real-valued random
variable we write $\bm{x}\in \mathbb{R}$, and similarly for other sets. The
total variation distance between two distributions $p$, $q$ is denoted by
$d_{\mathrm{TV}}(p,q)$.

\paragraph{Operator conventions.}
The $p$-norm $\norm{X}_p$ of an operator $X$ is the Schatten $p$-norm, and we
let $\norm{X}$ denote the Schatten $\infty$-norm of $X$, i.e., the operator
norm. We write $X\preceq Y$ if and only if $Y-X$ is positive semidefinite.

\subsection{Quasiprobability decompositions of quantum channels}\label{sec:qpd}
In this work, a QPD of a quantum channel $\calN: \mathsf{L}(\calH)\to
    \sfL(\calH)$ is\footnote{The current definition differs from some prior work in
    that the coefficients $a_i$ need not sum to one, because the $\calT_i$'s can
    introduce weights. This renders the moniker ``quasiprobability'' slightly
    misleading, though we keep the terminology to be consistent with other, more
    closely related, prior work, e.g., ~\cite{piveteau2023circuit}.} a
decomposition of the form
\begin{align}\label{eq:qpd_defn}
    \calN & = \sum_{i=1}^m a_i\ \calT_i \circ \calE_i = \norm{a}_1 \sum_{i=1}^m \frac{|a_i|}{\norm{a}_1} \mathrm{sign}(a_i)\ \calT_i \circ \calE_i.
\end{align}
\Cref{def:space_like_cut} and \Cref{def:time_like_cut} are special cases of this definition. The ingredients of the decomposition in \cref{eq:qpd_defn} are
\begin{itemize}
    %\item a positive integer $m$;
    \item a vector $a\in \mathbb{R}^m$;
    \item channels $\calE_i: \sfL(\calH)\to \sfL(\calH\otimes \calH_R)$ satisfying the
          desired constraint (i.e., an LOCC channel for space-like cuts, or a
          measure-and-prepare channel for time-like cuts); and
    \item \textit{post-processing operations} $\calT_i: \sfL(\calH\otimes \calH_R)\to\sfL(\calH)$ of the form
          $\calT_i : \rho \mapsto \Tr_R((\mathds{1}\otimes O_i)\rho)$, where $O_i\in\sfL(\calH_R)$ is some observable
          on the ancilla system $R$ satisfying $\norm{O_i}\leq 1$.
\end{itemize}

The key feature of decompositions of this form is that measuring the observable
$\norm{a}_1\ \mathrm{sign}(a_i) (O_i \otimes X)$ on the ensemble of states
$\{(|a_i|/\norm{a}_1,\ \calE_i(\rho))\}_i$ gives an unbiased estimator of
$\Tr(X\calN(\rho))$, for any $\rho\in \sfD(\calH)$ and Hermitian observable
$X\in \sfL(\calH)$. (Here, the index $i$ in the observable corresponds to the
value of $i$ that is drawn when randomly selecting the state in the ensemble to
prepare.)

For example, the local space-like cut we present in \Cref{sec:double_hadamard},
uses decompositions where $R=R_AR_B$ is a pair of ancilla qubits,
$O_i=\sigma_z\otimes \sigma_z$, and $\calE_i=\calV_i\otimes \calW_i$ are local
isometries in the ensemble. We then take the empirical mean of the outcomes
from measuring $\norm{a}_1\mathrm{sign}(a_i)\left[X\otimes
        (\sigma_z)_{R_A}\otimes (\sigma_z)_{R_B}\right]$ on the ensemble of states. For
the time-like cut presented in \Cref{sec:time_like}, we have $d_R=1$ (the
ancilla register is trivial), $O_i=1$, and the channels $\calE_i=\calM_i$ are
measure-and-prepare channels in the ensemble. In either case, taking the
empirical mean of $N$ trials
$\hat{\bm{\mu}}_1,\hat{\bm{\mu}}_2,\dots,\hat{\bm{\mu}}_N$ results in an
unbiased estimator $\hat{\bm{\mu}}$ with variance
$\Var[\hat{\bm{\mu}}]=\Var[\hat{\bm{\mu}}_1]/N$.

The definition of what constitutes a QPD given above contains as special cases
the definition given in Ref.~\cite{piveteau2023circuit} and the ``twisted
channel" construction in Ref.~\cite{zhao2023power}.

\paragraph{1-norm versus sample complexity.}
At first glance, an appropriate quantity to characterize the minimum value of
$N$ required for an accurate estimate would appear to be the 1-norm
$\norm{a}_1$. Indeed, since $\bm{\mu}_1$ clearly has magnitude at most
$\norm{a}_1$ with probability 1 (recall that $X$ is assumed to be bounded in
operator norm), taking $N$ be of the order of $\norm{a}_1^2$ suffices by a
straightforward application of Hoeffding's Inequality. This is the argument
provided in many if not all prior works on circuit cutting and related
applications of QPDs. For a space-like cut of a given channel $\calN$, the
minimum value of the 1-norm, $\gamma(\calN)$, can in turn be lower bounded by
examining the Choi state $J_{\calN}$ of $\calN$ and computing its
\textit{robustness of entanglement} $R(J_{\calN})$ (defined in
\Cref{sec:robustness}) using the results of Ref.~\cite{piveteau2023circuit}.
\begin{claim}[{Essentially~\cite[{Lemma~3.1}]{piveteau2023circuit}}]\label{claim:1_norm_lb}
    Let $\calN_{AB\to AB}$ be a bipartite quantum channel and consider space-like cuts of $\calN$ of the form in \Cref{def:space_like_cut}. It holds that
    \begin{align}
        \gamma(\calN) & \geq 1+2R(J_{\calN}).
    \end{align}
\end{claim}
We provide a self-contained proof in \Cref{sec:choi_robustness_bounds_1_norm_proof} for completeness.
A similar argument can be used to lower-bound the optimal 1-norm $\gamma_{\uparrow}(\calN)$ for time-like cuts
as well, as in~\cite[{Prop.~4.2}]{brenner2023optimal}. Intuitively, the less entangling an operation, the easier it should
be to replace using a QPD.

There is a danger, however, in assuming a number of samples of the order
$\gamma(\calN)^2$, or $\gamma_{\uparrow}(\calN)^2$, is also \textit{necessary}.
Firstly, as shown in \Cref{thm:rank_dep_wire_cutting}, this conclusion is
demonstrably false in some cases of practical interest: the variance in the
procedure associated with the optimal time-like cut scales at most like
$\gamma_{\uparrow}(\calN)$ for a class of non-trivial observables which
nevertheless satisfy $\norm{X}=1$. Moreover, bounding the 1-norm of the cut in
itself does not constitute an information-theoretic lower bound on the number
of samples required for a procedure of a similar spirit, but perhaps not
utilizing QPDs, to succeed. This raises the natural question: can we rigorously
prove that a QPD-based approach is sample-optimal in a non-trivial setting? We
answer this in the affirmative by showing in \Cref{thm:time_like_lower_bound}
that, for estimating output probabilities, any choice of measure-and-prepare
channels and classical post-processing in \Cref{proc:wire_cutting} requires the
same number of samples, up to a constant factor, as the QPD-based procedure we
give in \Cref{sec:good_mp_channels}.

\subsection{Diagonal 2-designs and the robustness of entanglement}\label{sec:robustness}
A \textit{diagonal $t$-design} on $n$ qubits is a unitary-operator-valued
random variable $\bm{U}\in\mathsf{U}(\mathbb{C}^{2^n})$ satisfying
\begin{align}
    \underset{\bm{U}}{\expct}\ \bm{U}^{\otimes t}X(\bm{U}^\dag)^{\otimes t} = \underset{\bm{\theta}}{\expct}\ V_{\bm{\theta}}^{\otimes t}X(V_{\bm{\theta}}^\dag)^{\otimes t}\quad \forall X\in\mathsf{L}(\mathbb{C}^{2^n})
\end{align}
where $\bm{\theta}=\bm{\theta}_{00\dots 0}\bm{\theta}_{00\dots 1}\dots\bm{\theta}_{11\dots 1}\in [0,2\pi)^{\{0,1\}^n}$ is uniformly random and $V_{\theta}\in\mathsf{U}(\mathbb{C}^{2^n})$ maps $\ket{x}$ to $\ee^{\ii \theta_x} \ket{x}$ for any $x\in\{0,1\}^n$. The implementation we make use of is stated in the following proposition. Here, a \textit{$k$-qubit phase-random circuit} is a circuit in which random diagonal $k$-qubit unitaries are applied to every possible combination of $n\choose k$ wires.
\begin{proposition}[Prop. 2 in~\cite{Nakata_2014}]
    For a system comprising $n$ qubits, a $k$-qubit phase-random circuit is an exact diagonal $t$-design if and only if $\min\{n,\ \lfloor \log t\rfloor + 1\}\leq k$.
\end{proposition}
In particular, in \Cref{sec:time_like} we employ the case $t=2$ in our procedure for implementing time-like cuts, for which random $2$-qubit diagonal gates suffice.

We will often have occasion to examine the robustness of entanglement of the
Choi state of the channels we consider. The robustness of
entanglement
%\footnote{Closely related to the exponential of the max-relative entropy of
    %entanglement~\cite{datta2009max}.}
$R(\rho_{AB})$ of a quantum state
$\rho_{AB}\in\calH_{AB}=\calH_A\otimes \calH_B$ is an entanglement measure
which quantifies the amount of ``mixing" with a separable state that is
required in order to bring $\rho_{AB}$ into the set of separable states. It is
defined through
\begin{align}\label{eq:pure_robustness_defn}
    R(\rho_{AB}) = \min\left\{s\geq 0: \frac{1}{1+s}\left(\rho_{AB}+s
    \sigma_-\right)\in \mathsf{Sep}(\calH_{AB}|A,B),\ \sigma_-\in
    \mathsf{Sep}(\calH_{AB}|A,B)\right\}.
\end{align}
In the case where $\rho_{AB}$ is pure a closed-form
characterization of the robustness in terms of its Schmidt
coefficients may be given~\cite{vidal1999robustness}:
\begin{align}\label{eq:pure_state_rob_expr}
    R(\outerprod{\psi}{\psi}) = (\sum_j \lambda_j)^2 - 1
\end{align}
where $\ket{\psi}\in \calH_A\otimes\calH_B$ has a Schmidt
decomposition
\begin{align}\label{eq:schmidt_decomp}
    \ket{\psi}_{AB} = \sum_j \lambda_j \ket{a_j}_A\otimes \ket{b_j}_B
\end{align}
with Schmidt coefficients $\lambda_j$ satisfying
$\lambda_1\geq \lambda_2\geq \dots\geq 0$. The
following theorem is a slight
modification of the construction that appears in~\cite{vidal1999robustness} which
will enable us to give an efficient construction of
good measure-and-prepare channels for the time-like cuts
we consider in this work.
\begin{theorem}[Similar to~\cite{vidal1999robustness}]\label{thm:pure_state_robustness}
    Suppose $\psi = \outerprod{\psi}{\psi}$ is a pure state of the form
    in \cref{eq:schmidt_decomp}. Let $R:=R(\psi)=
        (\sum_j\lambda_j)^2 - 1$ and $\bm{\theta}=\bm{\theta}_1\bm{\theta}_2
        \dots$ be a collection of independent and uniformly random angles
    $\bm{\theta}_j\in [0,2\pi)$.
    Define the random unit vectors
    \begin{align}
        \ket{u_{\bm{\theta}}} := \frac{1}{(1+R)^{1/4}}\sum_j\sqrt{\lambda_j}
        \ \ee^{i\bm{\theta}_j}\ket{a_j},\quad \ket{v_{\bm{\theta}}} :=\frac{1}{(1+R)^{1/4}}
        \sum_j\sqrt{\lambda_j}\ \ee^{i\bm{\theta}_j}\ket{b_j}
    \end{align}
    along with the separable states
    \begin{align}
        \sigma_- := \frac{1}{R}\sum_{k\neq \ell}\lambda_k\lambda_\ell \outerprod{a_k}{a_k}
        \otimes \outerprod{b_\ell}{b_\ell},\quad \sigma_+:=\expct_{\bm{\theta}}
        \outerprod{\overline{u_{\bm{\theta}}}}{\overline{u_{\bm{\theta}}}}
        \otimes\outerprod{v_{\bm{\theta}}}{v_{\bm{\theta}}}.
    \end{align}
    It holds that
    \begin{align}\label{eq:optimal_states}
        \sigma_+ = \frac{1}{1+R}\left(\psi + R\sigma_-\right).
    \end{align}
\end{theorem}
We say that two states $\sigma_+$, $\sigma_-$ which satisfy \cref{eq:optimal_states}
for a particular $\psi$ are \textit{optimal} for $\psi$. Though at first glance it may appear that $\sigma_+$ is defined as a convex combination
of an infinite family of difficult-to-implement pure states, we may straightforwardly take this to be
a more tractable, finite set in the following manner. Let $d=2^n$ and suppose $\bm{U}$ is a diagonal 2-design. Also, let $A_1,B_1,A_2,B_2$ be any unitaries which
have the following actions:
\begin{align}
     & A_1:\ket{1}_A\mapsto \frac{1}{(1+R)^{1/4}}\sum_j \sqrt{\lambda_j}\ket{j}_A,
     &                                                                             & B_1: \ket{1}_B\mapsto \frac{1}{(1+R)^{1/4}}\sum_j\sqrt{\lambda_j}\ket{j}_B\nonumber                                    \\
     & A_2:\ket{j}_A\mapsto \ket{a_j}_A,                                           &                                                                                     & B_2:\ket{j}_B\mapsto\ket{b_j}_B.
\end{align}
Then defining the random
unit vectors
\begin{align}
    \ket{s(\bm{U})}:= A_2 \bm{U}^\dag A_1 \ket{1}_A,\quad \ket{t(\bm{U})}:= B_2 \bm{U} B_1\ket{1}_B,
\end{align}
one may verify that
\begin{align}
    \expect{\bm{U}}\ \outerprod{s(\bm{U})}{s(\bm{U})}\otimes
    \outerprod{t(\bm{U})}{t(\bm{U})} & = \expect{\bm{\theta}}\
    \outerprod{\overline{u_{\bm{\theta}}}}{\overline{u_{\bm{\theta}}}}
    \otimes\outerprod{v_{\bm{\theta}}}{v_{\bm{\theta}}}.
\end{align}
Furthermore, as shown in~\cite{Nakata_2014}, a phase-random circuit $\bm{U}$
can be implemented by drawing 2-qubit gates from a
finite set of $6$ gates, independently and uniformly at random, at $O(n^2)$ fixed locations
in the circuit. In summary, one obtains an explicit description
of $\sigma_+$ as a random mixture of at most $6^{O(n^2)}$ pure states. Moreover,
if $A_1,B_1,A_2,$ and $B_2$ can be efficiently implemented, then the entire procedure to
prepare $\sigma_+$ is efficient in $n$.
\section{Space-like cuts}\label{sec:space_like}
In this section we introduce and analyze a procedure for local space-like cuts.
We first give some preliminary definitions in \Cref{sec:decomp} and
\Cref{sec:prod_extent} which will help us bound the cost of the procedure. We
then describe the procedure and bound its cost in \Cref{sec:double_hadamard}
before applying it to the problem of clustered Hamiltonian simulation in
\Cref{sec:hamiltonian_simulation}.

\subsection{Local decompositions of entangling unitaries}\label{sec:decomp}
Fix a bipartite system $AB$ with a corresponding Hilbert space
$\calH_{A}\otimes \calH_B$. A set $\Gamma$ is called a \emph{local
    decomposition} if it is of the form $ \Gamma = \{(c_i, V_i\otimes W_i): i\in
    [m]\} $ for some positive integer $m>0$, where for each $i\in [m]$ it holds
that $c_i\in\mathbb{R}$, $V_i\in\sfU(\calH_A)$, and $W_i\in\sfU(\calH_B)$. (We
omit the term ``local" whenever it is unambiguous and safe to do so.) We say
that a decomposition $\Gamma$ of the above form is \emph{valid} for a unitary
operator $U\in \sfU(\calH_{AB})$ if $U=\sum_{i\in [m]}c_iV_i\otimes W_i$. We
also define the \emph{magnitude} of such a decomposition $\Gamma$ as
\begin{align}
    \phi(\Gamma) := 2\norm{c}_1^2 - \norm{c}_2^2
\end{align}
viewing $(c_i:i\in [m])$ as a column vector in $\mathbb{R}^m$. Finally, we define the \emph{product} of two
decompositions $\Gamma_1=\{(a_i, V^{(1)}_i\otimes W^{(1)}_i):i\in [m_1]\}$ and $\Gamma_2=\{(b_i,V^{(2)}_i\otimes W^{(2)}_i):i\in [m_2]\}$ as
\begin{align}
    \Gamma_1\cdot \Gamma_2 := \{(a_ib_j, V^{(1)}_iV^{(2)}_j\otimes W^{(1)}_iW^{(2)}_j):(i,j)\in [m_1]\times [m_2]\}.
\end{align}
We then have the following straightforward but important observations.
\begin{lemma}\label{lem:prod_of_decomp_valid}
    Let $U,V\in \sfU(\calH_{AB})$. If $\Gamma_1$ and $\Gamma_2$ are valid local decompositions for $U$ and $V$, respectively, then $\Gamma_1\cdot \Gamma_2$ is a valid local decomposition for $UV$.
\end{lemma}
\begin{lemma}\label{lem:submult_of_mag}
    Let $\Gamma_1$ and $\Gamma_2$ be as defined above. It holds that
    \begin{align}
    \phi(\Gamma_1\cdot \Gamma_2)=2\norm{a}_1^2\norm{b}_1^2-\norm{a}_2^2\norm{b}_2^2\leq \phi(\Gamma_1)\phi(\Gamma_2)
    \end{align}
    where $a$ and $b$ are the coefficients in $\Gamma_1$ and $\Gamma_2$, respectively. Thus, the magnitude is submultiplicative.
\end{lemma}
\begin{proof}
    Define the column vector $c$ through $c_{ij}:= a_ib_j$. A straightforward algebra reveals that $\norm{c}_1^2 = \norm{a}_1^2\norm{b}_1^2$ and
    $\norm{c}_2^2 = \norm{a}_2^2\norm{b}_2^2$.  It suffices to show the inequality $2\norm{c}_1^2 - \norm{c}_2^2\leq (2\norm{a}_1^2 -
        \norm{a}_2^2)(2\norm{b}_1^2 - \norm{b}_2^2)$.
    Subtracting the left-hand side of the inequality from the right-hand side
    yields
    \begin{align}
        2(\norm{a}_1^2 - \norm{a}_2^2)(\norm{b}_1^2 - \norm{b}_2^2)
        \label{eq:submult}
    \end{align}
    which is nonnegative due to the inequality between norms $\norm{\cdot}_1\geq \norm{\cdot}_2$.
\end{proof}

For any nontrivial
decomposition, \cref{eq:submult} will be strictly positive. This is related to
the fact that cutting multiple gates together reduces the cost, which has been
previously observed in Refs.~\cite{schmitt2023cutting,piveteau2023circuit}.
Here we see additionally that the strict inequality holds by simply taking
products of the unitaries in the original decompositions.

\subsection{The product extent: an operational measure of entanglement}\label{sec:prod_extent}
Making use of the terminology established in the previous section we have the
following definition.
\begin{definition}[Product extent of a unitary]\label{defn:product_extent}
    The \emph{product extent} $\xi(U)$
    of a unitary operator $U\in\sfU(\calH_{AB})$ is defined as the minimum
    of $\phi(\Gamma)$ over all local decompositions $\Gamma$ which are
    valid for $U$.
\end{definition}
This definition relies on the fact that the minimum is always achieved, which in turn relies on the fact that optimal decompositions exist with a bounded number of terms. We prove these facts in \Cref{sec:prod_extent_well_defined}.
In addition to satisfying the desired
criteria for a ``dynamic" entanglement measure~\cite{nielsen2003dynamics},
the product extent is submultiplicative under composition. (See \Cref{lem:product_extent_entanglement_measure}.) This enables one to bound, for example, the product extent of a circuit from knowing the product extent of its individual gates.
The definition of the product extent is motivated by \Cref{thm:bound_on_overhead}, which implies that it is an achievable cost in a simple space-like cutting procedure. We repeat the theorem below for convenience, with a minor addition since we have now defined $\xi(U)$.

\theoremstyle{theorem}
\newtheorem*{T1}{\Cref{thm:bound_on_overhead}}
\begin{T1}[Rephrased]
    Let $U\in\sfU(\calH_{AB})$ be a unitary operation. For any local decomposition $\Gamma$ which is valid for $U$, the double Hadamard test of \Cref{sec:double_hadamard} yields a local space-like cut of $\calU\colon\rho\mapsto U\rho U^\dag$ with two ancilla qubits (i.e., $d_{R_A}=d_{R_B}=2$) and 1-norm $\phi(\Gamma)$. Moreover, if this decomposition is an operator Schmidt decomposition then
    \begin{align}
        \phi(\Gamma) = \xi(U) = \gamma(\calU) = 2\norm{c}_1^2 - 1.
    \end{align}
\end{T1}

The first part of the theorem follows by combining
\Cref{lem:gate_cutting_unbiased_estimator} and
\Cref{lem:gate_cutting_qpd_valid}. The second part follows from
\Cref{claim:1_norm_lb} and \Cref{prop:optimality_condn}. We now compare the
product extent to a previously introduced~\cite{harrow2003robustness}
entanglement measure for unitaries called the Choi-Jamiolkowski robustness.
This quantity is reminiscent of an entanglement measure for quantum states
$\rho$, the \textit{log-negativity}~\cite{plenio2005lognegativity} $\log
    \norm{\rho^\Gamma}_1$, and may be interpreted as an analogous measure for
quantum dynamics.  

\begin{definition}[Choi-Jamiolkowski robustness]\label{defn:choi_robustness}
    The \emph{Choi-Jamilkowski robustness} $R_c(U)$ of a bipartite unitary operator
    $U\in\sfU(\calH_A\otimes \calH_B)$ is defined as $1+2R(J_{\calU})$.
\end{definition}
Up to a constant shift and rescaling factor, this is just the robustness of entanglement of
the Choi state of the channel $\calU:\rho\mapsto U\rho U^\dag$. This is similar to how one might define robustness for channels relative to entanglement breaking channels except that we relax the constraint that the states in the mixture must be Choi states of valid channels.
The Choi-Jamiolkowski robustness may
be equivalently expressed as
\begin{align}
    \begin{aligned} R_c(U)=\min  \quad  & \norm{a}_1                           \\
                \textrm{s.t.} \quad & J_U=\sum_j a_j\rho_j\otimes \sigma_j \\
                                    & \rho_j\in \mathsf{D}(\calH_A)        \\
                                    & \sigma_j\in\mathsf{D}(\calH_B)
    \end{aligned}
\end{align}
by collecting terms appropriately.
The following proposition gives an easily
computable expression for the solution to this optimization problem.
\begin{proposition}\label{prop:choi_rob_formula}
    For any bipartite unitary operator $U\in\sfU(\calH_A\otimes \calH_B)$
    it holds that
    \begin{align}\label{eq:choi_robustness_unitary_special}
        R_c(U) & = 2(d_Ad_B)^{-1}{\lVert (UF)^\Gamma\rVert}_1^2 - 1.
    \end{align}
\end{proposition}
\begin{proof}
    Any operator $X\in\sfL(\calH_{A}\otimes\calH_B)$ has an operator Schmidt decomposition of
    the form
    \begin{align}
        X = \sum_j \lambda_j A_j\otimes B_j
    \end{align}
    where $\Tr(A_j^\dag A_k) = d_A \delta_{jk}$, $\Tr(B_j^\dag B_k) = d_B \delta_{jk}$,
    and $\sum_j \lambda_j^2 = 1$. Taking $U=X$, a simple algebra shows
    \begin{align}
        (UF)^{\Gamma} = \sum_j\lambda_j \vect(A_j)\vect(\overline{B_j})^\dag.
    \end{align}
    Using the orthogonality of the operators $A_j$ and $B_j$, we find
    that the singular vectors of $(UF)^\Gamma$ are proportional to
    $\vect(A_j)$ and $\vect(\overline{B}_j)$, and hence the singular values
    of $(UF)^\Gamma$ are $\sqrt{d_A d_B}\lambda_j$. Therefore, the right-hand side
    of \cref{eq:choi_robustness_unitary_special} is equal to $2(\sum_j \lambda_j)^2 - 1$.
    Moreover, it may be straightforwardly verified that $\lambda_j$ are the
    Schmidt coefficients of the pure state corresponding to the unit
    vector $(\mathds{1}_{A^\prime B^\prime}\otimes U_{AB})\ket{\Phi}_{A^\prime B^\prime AB}$. The claim then follows from \cref{eq:pure_state_rob_expr}
    and the fact that we are defining the Choi-Jamiolkowski robustness
    such that $R_c(U) = 1+2 R(J_U) = 1+2 ((\sum_j\lambda_j)^2 - 1) = 2(\sum_j\lambda_j)^2 -1$.
\end{proof}
Clearly, $R_c(U)\leq
    2d_Ad_B-1$, with equality if $U$ is a dual unitary operator such as a
SWAP operation. We can relate the product extent to the Choi-Jamiolkowski robustness as follows.
\begin{lemma}\label{lem:extent_robustness_bounds}
    For any bipartite unitary operator $U\in\sfU(\calH_A\otimes \calH_B)$
    it holds that
    \begin{align}
        1\leq R_c(U) \leq \xi(U)\leq 2 d_A^2d_B^2 - 1.
    \end{align}
\end{lemma}
\begin{proof}
    The first inequality is clear from the definition of the
    Choi-Jamiolkowski robustness. The final inequality can be seen by
    writing $U$ in the discrete Weyl (generalized Pauli) basis: if
    the vector $\alpha$ satisfies $U=\sum_{j=1}^{d_A^2 d_B^2}\alpha_j
        P_j\otimes Q_j$ for $P_j,Q_j$ discrete Weyl operators then
    \begin{align}
        1 = \frac{\Tr(U^\dag U)}{d_Ad_B} = \norm{\alpha}_2^2.
    \end{align}
    The maximum value of
    $\norm{\alpha}_1^2$ given this constraint is attained when
    $\alpha_j=1/d_A d_B$ for every $j$, so $\norm{\alpha}_1^2 \leq
        (d_Ad_B)^2$. It remains to prove that $R_c(U)\leq \xi(U)$. To this
    end, we show in the following paragraph that if $U=\sum_jc_j
        V_j\otimes W_j$ then a decomposition of the Choi state $J_U$ of the
    form $J_U=\sum_j a_j \rho_j\otimes \sigma_j$ exists with
    $\norm{a}_1 = 2\norm{c}_1^2-\norm{c}_2^2$, which
    proves the claim.

    Let
    \begin{align}
        \ket{a_j} := V_j \ket{\Phi}_{AA^\prime},\quad
        \ket{b_j}:=W_j\ket{\Phi}_{BB^\prime},
    \end{align}
    so that we may write
    \begin{align}
        J_U & = \sum_{ij}c_ic_j \outerprod{a_i}{a_j}\otimes \outerprod{b_i}{b_j} \\
            & = \sum_i c_i^2 \outerprod{a_i}{a_i}\otimes \outerprod{b_i}{b_i} +
        \sum_{i\neq j}c_ic_j \outerprod{a_i}{a_j}\otimes \outerprod{b_i}{b_j}.
    \end{align}
    We can decompose the ``off-diagonal" terms in the second sum using the following
    identity, which also appears in~\cite{eddins2022doubling} (and implicitly
    in~\cite{Mitarai2021overhead}):
    \begin{align}
        \sum_{p\in
            \mathbb{Z}_4}(-1)^p\outerprod{\alpha_{ij}^p}{\alpha_{ij}^p}\otimes
        \outerprod{\beta_{ij}^p}{\beta_{ij}^p} = \outerprod{a_i}{a_j}\otimes
        \outerprod{b_i}{b_j} + \outerprod{a_j}{a_i}\otimes \outerprod{b_j}{b_i}
    \end{align}
    where we have set
    \begin{align}
        \ket{\alpha_{ij}^p}:=\frac{1}{\sqrt{2}}\left(\ket{a_i}+\ii^p\ket{a_j}\right),\quad
        \ket{\beta_{ij}^p}=\frac{1}{\sqrt{2}}\left(\ket{b_i}+\ii^p\ket{b_j}\right).
    \end{align}
    Hence, we have
    \begin{align}
        J_U & = \sum_i c_i^2 \outerprod{a_i}{a_i}\otimes \outerprod{b_i}{b_i} +
        \frac{1}{2}\sum_{p\in
            \mathbb{Z}_4}(-1)^p\sum_{i\neq j}c_ic_j \outerprod{\alpha_{ij}^p}{\alpha_{ij}^p}\otimes
        \outerprod{\beta_{ij}^p}{\beta_{ij}^p}.
    \end{align}
    The sum of the absolute values of the coefficients in the above is equal to
    \begin{align}
        \norm{c}_2^2+2\sum_{i\neq j}|c_ic_j| & = 2\norm{c}_1^2 -
        \norm{c}_2^2
    \end{align}
    as desired.
\end{proof}

We conclude this discussion by showing that the product extent satisfies the desirable properties mentioned at the beginning of this section.

\begin{lemma}\label{lem:product_extent_entanglement_measure}
    The product extent
    satisfies: \begin{enumerate}[label=\roman*)] \item \emph{Faithfulness}:
              $\xi(U)=1$ iff $U$ is a product of local unitaries.  \item \emph{Local
                  unitary invariance}: $\xi((V_A\otimes V_B) U (W_A\otimes W_B)) =
                  \xi(U)$.  \item \emph{Submultiplicativity}: $\xi(UV)\leq \xi(U)\xi(V)$.
    \end{enumerate} \end{lemma} \begin{proof} To show faithfulness of $\xi$, we
    make use of \Cref{lem:extent_robustness_bounds} along with the fact
    that $R_c(U)$ is faithful, from which it follows that if $\xi(U)=1$ then $U$
    must be a product of local unitaries. The other direction is
    straightforward. Local unitary invariance follows from the local unitary
    invariance of the feasible set in the optimization problem which defines
    $\xi$. Submultiplicativity follows from the definition of $\xi$ since the magnitude $\phi$ is submultiplicative.
\end{proof}

\subsection{Sufficient conditions for optimality of the local space-like cut}
By \Cref{claim:1_norm_lb} we have that $R_c(U)$ bounds from below the 1-norm in
any space-like cut of $\calU\colon \rho\to U\rho U^\dag$. Hence, we say that a
space-like cut of $\calU$ into LOCC channels whose 1-norm saturates this bound
is optimal. Remarkably, in many cases our procedure --- which only makes use of
local unitary operations and does not use classical communication --- achieves
this notion of optimality, which addresses an open question from
Ref.~\cite{piveteau2023circuit}.
\begin{proposition}\label{prop:optimality_condn}
    Suppose $U\in\sfU(\calH_A\otimes \calH_B)$
    admits an operator Schmidt decomposition whose
    Schmidt operators are each proportional to some unitary
    operator. Then $\xi(U) = R_c(U)$.
\end{proposition}
\begin{proof}
    By \Cref{lem:extent_robustness_bounds}, we have $\xi(U)\geq R_c(U)$.
    Let $U=\sum_j \lambda_j V_j\otimes W_j$ be the operator Schmidt
    decomposition of $U$, such that $\Tr(V_j^\dag V_k)=d_A \delta_{jk}$
    and $\Tr(W_j^\dag W_k)=d_B \delta_{jk}$, and $\sum_j\lambda_j^2 = 1$,
    and $V_j$, $W_j$ are unitary. Then
    \begin{align}
        \xi(U) & \leq 2(\sum_j\lambda_j)^2 - \sum_j\lambda_j^2 \\
               & = 2(\sum_j\lambda_j)^2-1                      \\
               & = 1+2R(J_U)                                   \\
               & =R_c(U).\qedhere
    \end{align}
\end{proof}
We remark that 2-qubit gates, generalized SWAP operations, products of transversal 2-qubit gates, and certain
controlled-Pauli operations all fall into the required category of unitary.

\subsection{The double Hadamard test}\label{sec:double_hadamard}

In this section we describe the procedure which leads to
\Cref{thm:bound_on_overhead}. In addition to yielding the bound in the
statement of the theorem, this procedure allows us to give explicit
descriptions of the required products of local unitaries in the relevant QPD.

\begin{figure}
    \centering
    \begin{subfigure}[c]{0.28\textwidth}
        \includegraphics[width=\textwidth]{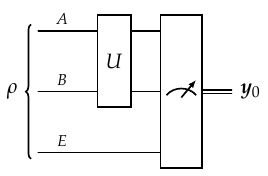}
    \end{subfigure}
    \hspace{2em}\begin{large}$\longleftrightarrow$\end{large}\hspace{2em}
    \begin{subfigure}[c]{0.53\textwidth}
        \includegraphics[width=\textwidth]{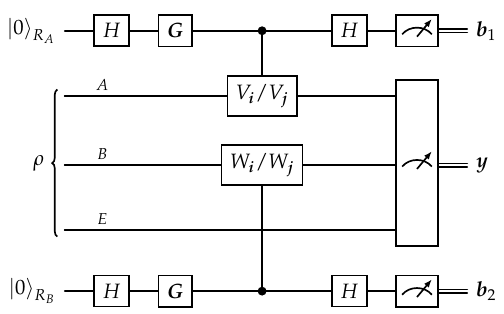}
    \end{subfigure}
    \caption{\label{fig:double_hadamard_test} \textbf{Left}: a bipartite unitary operation $U$ acting on a subsystem $AB$ of the system $ABE$ followed by measurement of an observable $X$, yielding an outcome $\bm{y}_0\in \mathrm{spec}(X)$. \textbf{Right}: an ensemble of
        ``double Hadamard test''
        circuits for space-like cutting of the unitary making use of the decomposition
        $U=\sum_i c_i V_i\otimes W_i$. The ensemble is generated
        by the random variables $\bm{i}$, $\bm{j}$, and $\bm{G}$,
        as described in \Cref{sec:double_hadamard}. The estimator
        described in \Cref{lem:gate_cutting_unbiased_estimator} is
        constructed from the measurement outcomes from these random
        circuits, and bounding the variance of this estimator leads to
        \Cref{thm:bound_on_overhead}.}
\end{figure}

Fix a positive integer $m\in \mathbb{Z}$ as well as a valid local decomposition
$\Gamma$ for the unitary $U\in\mathsf{U}(\calH_{AB})$ of the form
$\Gamma=\{(c_i, V_i\otimes W_i):i\in [m]\}$, such that $U=\sum_{i=1}^m c_i
    V_i\otimes W_i$ and $c_i>0$ for each $i\in [m]$. Note that this positivity
requirement is without loss of generality compared to the decompositions
appearing previously since the sign of each $c_i$ can be absorbed into the
unitary operators. We define the \emph{setting random variables}
$\bm{i},\bm{j}\in [m]$ and $\bm{g}\in \{0,1\}$ with joint probability mass
function (PMF) given by
\begin{align}\label{eq:algo_prob_dist}
    p(i,j,g) =
    \begin{cases}
        0                       & \text{if}\ i=j\ \text{and}\ g=1 \\
        c_ic_j\phi(\Gamma)^{-1} & \text{otherwise}
    \end{cases}
\end{align}
Note that this is a valid PMF since
\begin{align}
    \sum_{ijg} p(i,j,g) = \phi(\Gamma)^{-1}\left(\sum_i c_i^2 + 2\sum_{i\neq j}c_ic_j\right) = \phi(\Gamma)^{-1}(2\norm{c}_1^2-\norm{c}_2^2) = 1.
\end{align}
Using the setting random variables, we define a corresponding
\textit{random} local unitary circuit acting on $AB$ and a
pair of local ancilla qubits $R_AR_B$ according to \Cref{fig:double_hadamard_test}. These circuits have a nearly
identical form to that for two simultaneous Hadamard tests
except all qubits are measured. Here, $\bm{G}=\mathds{1}$
if $\bm{g}=0$ and $\bm{G}=S$ (a single-qubit phase gate) otherwise, and
\begin{align}
    \begin{tikzcd}[ampersand replacement=\&]
        \& \ctrl{1} \& \qw \\
        \& \gate{V_{\bm{i}}/V_{\bm{j}}} \& \qw
    \end{tikzcd} & = \outerprod{0}{0}\otimes V_{\bm{i}} + \outerprod{1}{1}\otimes V_{\bm{j}},
    \quad
    \begin{tikzcd}[ampersand replacement=\&]
        \& \ctrl{1} \& \qw \\
        \& \gate{W_{\bm{i}}/W_{\bm{j}}} \& \qw
    \end{tikzcd} & = \outerprod{0}{0}\otimes W_{\bm{i}} + \outerprod{1}{1}\otimes W_{\bm{j}}.
\end{align}
The following lemma implies that this ensemble of operations allows one to estimate an expectation value with respect to the output of the original circuit using a number of samples at most on the order of $\phi(\Gamma)^2\norm{X}^2$, where $X$ is the observable of interest.
\begin{lemma}\label{lem:gate_cutting_unbiased_estimator}
    Let $X\in\sfL(\calH_{ABE})$ be a Hermitian observable and $\rho\in\sfD(\calH_{ABE})$ be a quantum state. Define $\bm{y}\in \mathrm{spec}(X)$ and $\bm{b}=\bm{b}_1\bm{b}_2\in\{0,1\}^2$ to be the random variables obtained from measuring $X$ on the register $ABE$ and measuring $R_AR_B$ in the computational basis, respectively, on the output of the random circuit in \Cref{fig:double_hadamard_test}. It holds that $\hat{\bm{\mu}}:= \phi(\Gamma) (-1)^{\bm{g}+\bm{b}_1+\bm{b}_2}\ \bm{y}$
    is an unbiased estimator of $\mu:=\Tr(X U_{AB} \rho_{ABE} U_{AB}^\dag)$.
\end{lemma}
If $\norm{X}\leq 1$ and we pick an optimal $\Gamma$ so that $\phi(\Gamma)=\xi(U)$ then the estimator has variance at most $\xi(U)^2$.

\subsubsection*{Proof of \Cref{lem:gate_cutting_unbiased_estimator}}
By linearity, it suffices to show that the lemma holds for any
pure input state $\rho_{ABE} = \outerprod{\psi}{\psi}_{ABE}$. Let $\mathcal{B}(i,j,g)$
denote the event where $\bm{i}=i$, $\bm{j}=j$, and $\bm{g}=g$ for any
value of the setting random variables $(i,j,g)\in [m]^2\times \{0,1\}$.
We may write
%\begin{samepage}
\begin{align}\label{eq:expect_decomp}
    \expct \hat{\bm{\mu}} & = \sum_{i}c_i^2\expct\left[(-1)^{\bm{b}_1+\bm{b}_2} \bm{y}|\mathcal{B}(i,i,0)\right] \nonumber
    \\ &\qquad\qquad + \sum_{i\neq j}c_ic_j\left(\expct\left[(-1)^{\bm{b}_1+\bm{b}_2} \bm{y}|\mathcal{B}(i,j,0)\right] - \expct\left[(-1)^{\bm{b}_1+\bm{b}_2}\bm{y}|\mathcal{B}(i,j,1)\right]\right).
\end{align}
%\end{samepage}
Conditioned on $\mathcal{B}(i,i,0)$, the state (see
\Cref{fig:double_hadamard_test}) prior to measurement is $ \ket{0}_{R_A}
    (V_i\otimes W_i)_{AB}\ket{\psi}_{ABE} \ket{0}_{R_B}. $ (Here and for the
remainder of this section we use implicit identities acting on the register
$E$.) Hence, $(\bm{b}_1,\bm{b}_2)=(0,0)$ with probability 1 and the first sum
in \cref{eq:expect_decomp} is equal to
\begin{align}
    \sum_ic_i^2\Tr(X (V_i\otimes W_i)\outerprod{\psi}{\psi}(V_i^\dag\otimes W_i^\dag)).
\end{align}
A preview of the conclusion of the next part of the argument is as follows: for any
$i\neq j$, the $(i,j)^{\text{th}}$ term in the second sum is equal to
\begin{align}\label{eq:neq_term}
    \frac{c_ic_j}{2}\left[\Tr\left(X(V_i\otimes W_i)\outerprod{\psi}{\psi}(V_j^\dag\otimes W_j^\dag)\right) + \Tr\left(X(V_j\otimes W_j)\outerprod{\psi}{\psi}(V_i^\dag\otimes W_i^\dag)\right)\right]
\end{align}
which implies that the right-hand
side of \cref{eq:expect_decomp}  is equal to
\begin{align}
    \sum_{i,j\in [m]} c_ic_j\Tr(X(V_i\otimes W_i)\outerprod{\psi}{\psi}(V_j^\dag\otimes W_j^\dag)) = \Tr(X U\outerprod{\psi}{\psi}U^\dag)
\end{align}
as required. Let us now show this.
\begin{claim}\label{claim:expect_diff}
    For any $i,j\in [m]$ with $i\neq j$ it holds that
    \begin{align}\label{eq:expect_diff}
         & \expct\left[(-1)^{\bm{b}_1+\bm{b}_2} \bm{y}|\mathcal{B}(i,j,0)\right] - \expct\left[(-1)^{\bm{b}_1+\bm{b}_2}\bm{y}|\mathcal{B}(i,j,1)\right]\nonumber
        \\ &\qquad\quad = \frac{1}{2}\left[\Tr(X (V_i\otimes W_i)\outerprod{\psi}{\psi}(V_j^\dag\otimes W_j^\dag)) + \Tr(X (V_j\otimes W_j)\outerprod{\psi}{\psi}(V_i^\dag\otimes W_i^\dag))\right].
    \end{align}
\end{claim}
\begin{proof}
    Let us consider the case where $i=1$ and $j=2$ for notational clarity:
    the other cases are identical. Define the states
    \begin{align}
         & \ket{\psi_{00}}:= (V_1\otimes W_1)\ket{\psi} & \ket{\psi_{01}}:= (V_1\otimes W_2)\ket{\psi}\nonumber \\
         & \ket{\psi_{10}}:= (V_2\otimes W_1)\ket{\psi} & \ket{\psi_{11}}:= (V_2\otimes W_2)\ket{\psi}
    \end{align}
    Conditioned on the event $\mathcal{B}(1,2,0)$,
    the circuit in \Cref{fig:double_hadamard_test} acts as
    \begin{align}
        \ket{00}_{R_A R_B}\ket{\psi}_{AB} & \longrightarrow \frac{1}{2}\sum_{a\in \{0,1\}^2}\ket{a}_{R_A R_B}\otimes \ket{\psi_{a_1a_2}}_{AB}                       \\
                                          & \longrightarrow \frac{1}{4}\sum_{a,a^\prime\in\{0,1\}^2}(-1)^{a\cdot a^\prime}\ket{a^\prime}\otimes \ket{\psi_{a_1a_2}}
    \end{align}
    where $a\cdot b := a_1b_1 + a_2b_2$. Therefore, the probability
    of observing the outcome $y$ from measuring $AB$ according to $X$ and
    $b=b_1b_2$ from measuring $R_A R_B$ in the computational basis
    is equal to
    \begin{align}\label{eq:real_part_prob}
        \frac{1}{16}\sum_{a,a^\prime\in\{0,1\}^2} (-1)^{(a-a^\prime)\cdot b} \langle \psi_{a_1^\prime a_2^\prime}| \Pi_y |\psi_{a_1 a_2}\rangle
    \end{align}
    in this case. Similarly, if $\mathcal{B}(1,2,1)$ occurs, the circuit
    produces the state
    \begin{align}
        \frac{1}{4}\sum_{a,a^\prime\in\{0,1\}^2}\ii^{a \cdot(1,1)} (-1)^{a\cdot a^\prime}\ket{a^\prime}\otimes \ket{\psi_{a_1a_2}}.
    \end{align}
    such that the probability of observing the outcomes $b=b_1b_2$ and $y$
    from measuring $R_A R_B$ and $AB$, respectively, is
    \begin{align}\label{eq:imag_part_prob}
        \frac{1}{16}\sum_{a,a^\prime\in\{0,1\}^2} \ii^{(a-a^\prime)\cdot (1,1)}(-1)^{(a-a^\prime)\cdot b} \langle \psi_{a_1^\prime a_2^\prime}| \Pi_y |\psi_{a_1 a_2}\rangle.
    \end{align}
    Using \cref{eq:real_part_prob} and \cref{eq:imag_part_prob} we find that
    the left-hand side of \cref{eq:expect_diff} (with $i=1$ and $j=2$)
    is equal to
    \begin{align}\label{eq:long_expect_diff}
        \frac{1}{16}\sum_{\substack{y\in \mathrm{spec}(X)         \\ b\in\{0,1\}^2\\ a,a^\prime \in \{0,1\}^2}}(-1)^{b\cdot (1,1)} y &\left((-1)^{(a-a^\prime)\cdot b}-\ii^{(a-a^\prime)\cdot (1,1)}(-1)^{(a-a^\prime)\cdot b}\right)\langle \psi_{a_1^\prime a_2^\prime}| \Pi_y |\psi_{a_1 a_2}\rangle\nonumber\\
         & = \frac{1}{16}\sum_{\substack{a,a^\prime \in \{0,1\}^2 \\ b\in\{0,1\}^2}} (-1)^{(a-a^\prime + (1,1))\cdot b}\left(1-\ii^{(a-a^\prime)\cdot (1,1)}\right)\langle \psi_{a_1^\prime a_2^\prime}| X |\psi_{a_1 a_2}\rangle.
    \end{align}
    A straightforward case analysis shows
    \begin{align}
        \sum_{b\in\{0,1\}^2}(-1)^{(a-a^\prime + (1,1))\cdot b}\left(1-\ii^{(a-a^\prime)\cdot (1,1)}\right) = \begin{cases}
                                                                                                                 8 & \textnormal{if}\ a=00\ \textnormal{and} \ a^\prime=11 \\
                                                                                                                 8 & \textnormal{if}\ a=11\ \textnormal{and}\ a^\prime=00  \\
                                                                                                                 0 & \textnormal{otherwise}
                                                                                                             \end{cases}
    \end{align}
    from which we may conclude that the right-hand side of \cref{eq:long_expect_diff}
    is equal to
    \begin{align}
        \frac{1}{2}\left(\langle \psi_{00}| X | \psi_{11}\rangle + \langle \psi_{11}| X | \psi_{00}\rangle\right).
    \end{align}
    The claim follows from the definitions of $\ket{\psi_{00}}$ and $\ket{\psi_{11}}$.
\end{proof}

The next lemma completes the proof of the first part of
\Cref{thm:bound_on_overhead}, concerning the existence of a space-like cut of
the desired form.
\begin{lemma}\label{lem:gate_cutting_qpd_valid}
    \Cref{lem:gate_cutting_unbiased_estimator} implies that the channel
    $\calU\colon \rho\mapsto U\rho U^\dag$ has a space-like cut of the form described in \Cref{thm:bound_on_overhead}.
\end{lemma}
\begin{proof}
    Let $E$ be a copy of the system $AB$. By \Cref{lem:gate_cutting_unbiased_estimator} we have
    \begin{align}\label{eq:unbiased_1}
        \expct \hat{\bm{\mu}} & = \Tr(X(\id_E\otimes \calU)(\rho))
    \end{align}
    for any Hermitian observable $X\in\sfL(\calH_{ABE})$ and state $\rho\in\sfD(\calH_{ABE})$, where $\hat{\bm{\mu}}$ is defined as in the lemma. Using \cref{eq:expect_decomp}, the definitions of $\bm{y}$ and $\bm{b}_1\bm{b}_2$, and making use of \Cref{fig:double_hadamard_test}, the expected value of $\hat{\bm{\mu}}$ can be written
    \begin{align}
        \expct \hat{\bm{\mu}} & = \sum_{i,j}c_ic_j\Tr\left((X\otimes (\sigma_z\otimes\sigma_z)_{R_A R_B})(\id_E\otimes\calV^{(i,j,0)}_{AR_A}\otimes \calW^{(i,j,0)}_{BR_B})(\rho\otimes \outerprod{00}{00}_{R_AR_B})\right) \nonumber
        \\ &\qquad\qquad - \sum_{i\neq j}c_ic_j\Tr\left((X\otimes (\sigma_z\otimes\sigma_z)_{R_A R_B})(\id_E\otimes \calV^{(i,j,1)}_{AR_A}\otimes \calW^{(i,j,1)}_{BR_B})(\rho\otimes \outerprod{00}{00}_{R_AR_B})\right)\\
                              & =: \Tr\left(X(\id_E\otimes \widetilde{\calU})(\rho)\right)\label{eq:utilde_defn}
    \end{align}
    where $\calV^{(i,j,g)}$ and $\calW^{(i,j,g)}$ denote the actions of the local
    circuits on the subsystems $AR_A$ and $BR_B$, respectively, in \Cref{fig:double_hadamard_test},
    conditioned on the event where $\bm{i}=i$, $\bm{j}=j$, and $\bm{g}=g$, and in the second line we have defined the map $\widetilde{\calU}:\sfL(\calH_{AB})\to\sfL(\calH_{AB})$ by
    \begin{align}
        \widetilde{\calU}(Y_{AB}) & = \sum_{i,j}c_i c_j \Tr_{R_A R_B}\left((\mathds{1}_{AB}\otimes \sigma_z\otimes \sigma_z)(\calV^{(i,j,0)}_{AR_A}\otimes \calW^{(i,j,0)}_{BR_B})(Y_{AB}\otimes \outerprod{00}{00}_{R_AR_B})\right)\nonumber
        \\ &\qquad\qquad - \sum_{i\neq j}c_i c_j \Tr_{R_A R_B}\left((\mathds{1}_{AB}\otimes \sigma_z\otimes \sigma_z)(\calV^{(i,j,1)}_{AR_A}\otimes \calW^{(i,j,1)}_{BR_B})(Y_{AB}\otimes \outerprod{00}{00}_{R_AR_B})\right)
    \end{align}
    for all $Y_{AB}\in \sfL(\calH_{AB})$. By inspection, if $\calU=\widetilde{\calU}$ then the right-hand side of the above is a QPD of the desired form. It remains to show that $\calU=\widetilde{\calU}$. But this is clear from the fact we may pick a basis of Hermitian observables for the vector space $\sfL(\calH_{ABE})$ and use \cref{eq:unbiased_1} and \cref{eq:utilde_defn} for each element of this basis to conclude that
    \begin{align}
        (\id_E\otimes\widetilde{\calU})(\rho) = (\id_E\otimes\calU)(\rho)
    \end{align}
    for any $\rho\in\sfD(\calH_{ABE})$. Picking $\rho=\Phi_{AB}$ to be the maximally entangled state and noting that the function taking linear maps to their Choi states is a bijection proves the claim.
\end{proof}

\subsection{Application to clustered Hamiltonian simulation}\label{sec:hamiltonian_simulation}
Our procedure for space-like cutting is applicable to the simulation of
clustered quantum systems, as previously considered
in~\cite{peng2020simulating,Childs_2021_trotter_error}. Unlike the setting
introduced in these works, we focus on partitioning the system into just two
disjoint subsets, and we also assume the local terms in the Hamiltonian of the
system are proportional to Pauli strings. However, our setting is more general
in other ways: we do not require \emph{geometric} locality, nor a restriction
to 2-local interactions between qubits in a bounded-degree interaction graph.
Instead, we consider systems of $n$ qubits whose Hamiltonian we take to be of
the general form
\begin{align}\label{eq:clustered_hamiltonian}
    H=\sum_{i\in E_A}H_i + \sum_{j\in E_B}H_j + \sum_{k\in \partial A}H_k
\end{align}
where $A,B\subset [n]$ is a partition of the $n$ qubits into
disjoint subsystems comprising $n_A$ and $n_B$ qubits, respectively,
each term is $O(1)$-local (acts non-trivially on at most $O(1)$ subsystems),
and terms $H_i$ with $i\in E_A$ ($i\in E_B$) act non-trivially only on qubits in $A$ ($B$), while
terms $H_k$ with $k\in \partial A$ act non-trivially on qubits in both
subsystems. We also assume that each term satisfies $\norm{H_i}\leq 1$ and the
total number of terms is at most $\text{poly}(n)$. Suppose
further that there is an observable of interest $X$ such that $\norm{X}\leq 1$ which can be
efficiently measured and whose eigenvectors are
product states with respect to $AB$, e.g., computational basis
measurements and efficient post-processing, or Pauli observables
of the form $X=X_A\otimes X_B$. The following is
then a formal description of the task considered in prior work
in the case of bipartitioning.

\begin{tcolorbox}[colback=white, arc=0pt, boxrule=0.5pt]
    \begin{problem}[Clustered Hamiltonian Simulation]\label{prob:clustered_ham_sim}$ $\vspace{0.5em}\newline
    \textbf{Input}: $N$ copies of some initial state $\rho_{AB}$, an accuracy parameter $\veps > 0$, a simulation time $t\in \mathbb{R}$,
    and classical descriptions of i) a Hamiltonian $H$ of the form in \cref{eq:clustered_hamiltonian} and ii) an observable $X$ of the form described above.\vspace{0.5em}\\
    \textbf{Output}: An estimate $\hat{\bm{\mu}}$ of $\mu:=\Tr(X\ee^{-\ii H t} \rho_{AB}\ee^{\ii H t})$ s.t.\ $|\bm{\hat{\mu}}-\mu|\leq \veps$ with
    high probability.
    \end{problem}
\end{tcolorbox}
In the statement of the following theorem, ``polynomial-size"
indicates circuits which are of size
polynomial in $n$, $t$, and $\veps$, while
``locally'' refers to product unitaries with respect to subsystems $AR_A$
and $BR_B$.
\begin{theorem}\label{thm:clustered_ham_sim}
  \Cref{prob:clustered_ham_sim} can be solved using
  a quantum algorithm which computes $\hat{\bm{\mu}}$ using efficient classical post-processing of the measurement outcomes obtained from $N=O(\ee^{4 \eta t}/\veps^2)$ independent executions of random, polynomial-size quantum circuits each acting \emph{locally} on a copy of $\rho_{AB}\otimes \outerprod{00}{00}_{R_A R_B}$, where $\eta:=\sum_{k\in \partial A}\norm{H_k}$
  and $R_A$ and $R_B$ are a pair of ancilla qubits. 
    Moreover, there is a classical algorithm to sample these circuits in time $\mathrm{poly}(n,t,1/\veps)$.
\end{theorem}
It is straightforward to see that we may execute the circuits in this theorem
one-at-a-time on a single system of $\max\{n_A,n_B\}+1$ qubits so long as the initial state is
a product state $\rho=\rho_A\otimes \rho_B$, making use of the assumption that
the observable $X$ is implementable without applying entangling unitaries prior
to measuring. This allows one to recover the
originally suggested use-case of solving this problem from Ref.~\cite{peng2020simulating}; that is, reducing the number of qubits required
to perform a quantum simulation task. Additionally, a more precise bound on the variance of $\hat{\bm{\mu}}$ is possible. Namely, the proof relies on a local decomposition $\Gamma$ of the unitary circuit arising from $r$ Trotter steps having magnitude
\begin{align}\label{eq:gamma_magnitude}
    \phi(\Gamma)\leq 2\ee^{2\eta t}\left(1+ \frac{4\eta^2 t^2}{r}\right)-1.
\end{align}
Thus, we see that the multiplicative overhead incurred by the procedure is close to one for weak interactions.
\begin{proof}[Proof of \Cref{thm:clustered_ham_sim}]
    We describe an algorithm which satisfies all the required properties. The algorithm is based on a straightforward local decomposition of the first-order Trotter formula
    for the time-evolution $\mathscr{T}:\mathbb{R}\to \sfU(\calH_{AB})$ defined for all $x\in\mathbb{R}$ as
    \begin{align}
        \mathscr{T}(x) = \prod_{j\in E_A \cup E_B} \ee^{-\ii H_j x}\prod_{k\in \partial A} \ee^{-\ii H_k x}.
    \end{align}
    By \cite[{Corollary~12}]{Childs_2021_trotter_error}, for instance, taking a sufficiently large positive integer $r=O(\mathrm{poly}(n)\ t^2\ \veps^{-1})$ and defining $U:=\mathscr{T}(t/r)^r$ we have
    \begin{align}
        \lVert U-\ee^{-\ii H t}\rVert\leq \veps/4.
    \end{align}
    By this choice, making use of the assumption that $\norm{X}\leq 1$ it holds that
    \begin{align}
        |\Tr(X \ee^{-\ii H t} \rho \ee^{\ii H t})-\Tr(X U\rho U^\dag)|\leq \veps/2.
    \end{align}
    Hence, to solve \Cref{prob:clustered_ham_sim}, i.e., output an $\veps$-accurate estimate of the expectation $\Tr(X\ee^{-\ii H t}\rho \ee^{\ii H t})$, it suffices to produce an estimate for $\Tr(X U \rho U^\dag)$ which is accurate to within $\veps/2$. We may accomplish this using the procedure in \Cref{thm:bound_on_overhead} in the following manner. We consider local decompositions of the unitary operation $\ee^{-\ii H_j t / r}$. If $j\in E_A\cup E_B$ then this operator itself is a local decomposition via the singleton set $\Gamma_j = \{(1,\ \ee^{-\ii H_j t /r})\}$ since $\ee^{-\ii H_j t /r}$ is of the form $U_A\otimes U_B$. If instead $j\in \partial A$ we use the local decomposition
    \begin{align}
        \Gamma_j = \{(\cos(\lVert H_j\rVert t/r),\ \mathds{1}\otimes \mathds{1}),\ (\sin(\lVert H_j \rVert t/r),\ -\ii \ H_j/\rVert H_j \lVert)\}.
    \end{align}
    where here and throughout the proof we set $c_{j,0}:= |\cos(\lVert H_j\rVert t/r)|$ and $c_{j,1}:= |\sin(\lVert H_j\rVert t/r)|$ for each $j\in \partial A$. By \Cref{lem:submult_of_mag} the local decomposition
    \begin{align}\label{eq:gamma_defn_ham_sim}
        \Gamma:=\left(\prod_{j\in E_A\cup E_B}\Gamma_j\prod_{k\in\partial A}\Gamma_k\right)^r
    \end{align}
    of $U$ has magnitude
    \begin{align}
        \phi(\Gamma)&= \prod_{k\in\partial A}(|c_{k,0}|+|c_{k,1}|)^{2r} - \prod_{k\in \partial A}(c_{k,0}^2+c_{k,1}^2)^r\\
        &=\prod_{k\in\partial A}(|c_{k,0}|+|c_{k,1}|)^{2r} - 1.\label{eq:mag_expression}
    \end{align}
    Next, let $\eta_2:=\sqrt{\sum_{k\in\partial A}\norm{H_k}^2}$. Then so long as $r$ is at least $t$ it holds that
    \begin{align}
        \left|\prod_{k\in\partial A}(|c_{k,0}|+|c_{k,1}|)^{2r} - \ee^{2\eta t}\right|&= \ee^{2\eta t} \left|\prod_{k\in\partial A}\left(\frac{c_{k,0}+c_{k,1}}{\ee^{\norm{H_k}t/r}}\right)^{2r} - 1\right|\\
        &\leq 2r\ee^{2\eta t}\sum_{k\in \partial A}\left|\frac{c_{k,0}+c_{k,1}}{\ee^{\norm{H}_k t/r}}-1\right|\\
        &\leq 2r\ee^{2\eta t}\sum_{k\in \partial A}\left|c_{k,0}+c_{k,1}-\ee^{\norm{H}_k t/r}\right|\\ 
        &\leq 4\ee^{2\eta t}t^2\eta_2^2/r
    \end{align}
    where in the first line we used the definition of $\eta$ as well as the fact that $\cos(x)$ and $\sin(x)$ are nonnegative for $x\in [0,1]$, in the second line we used the bound
    \begin{align}
        \left|\prod_{i=1}^n a_i -\prod_{i=1}^nb_i\right|\leq \sum_{i=1}^n | a_i - b_i|
    \end{align}
    whenever $a_i,b_i\in [0,1]$, and in the final line we made use of the inequalities $0\leq \ee^x-\cos(x)-\sin(x)\leq 2x^2$ for every $x\in [0,1]$. This implies that
    \begin{align}
        \phi(\Gamma)\leq 2\ee^{2\eta t}\left(1+\frac{4t^2\eta_2^2}{r}\right) - 1
    \end{align}
    for $r$ at least $t$. The bound in Eq.~\eqref{eq:gamma_magnitude} follows from the inequality $\eta_2\leq \eta_1$. Furthermore, if $r$ is at least, say, $100 \eta^2t^2$, then $\phi(\Gamma)=O(\ee^{2\eta t})$. The bound in \Cref{thm:bound_on_overhead} then implies that we can use the procedure in \Cref{sec:double_hadamard} with the decomposition $\Gamma$ to estimate $\Tr(X U \rho U^\dag)$ to within accuracy $\veps/2$ using $O(\ee^{4\eta t}/\veps^2)$ independent executions of the random circuits described therein, which act locally on $A R_A$ and $B R_B$ as required. It remains to show that this procedure, including the classical post-processing step, can be implemented efficiently. Since $X$ is efficiently measurable by assumption, and the procedure in \Cref{sec:double_hadamard} involves computing the empirical mean of the observable
    $\phi(\Gamma)\left[(\sigma_z)_{R_A}\otimes (\sigma_z)_{R_B}\otimes X_{AB}\right]$ on the output of some random circuits, it suffices to show that we can efficiently i) compute $\phi(\Gamma)$, ii) implement a given circuit from the ensemble, and iii) sample from the appropriate distribution for $\Gamma$, as described in \cref{eq:algo_prob_dist}.

    The first claim is evident from \cref{eq:mag_expression}, since it is a product of polynomially-many terms.
    
    The second claim holds since, by inspecting \cref{eq:gamma_defn_ham_sim} and
    making use of the definitions of the $\Gamma_j$ terms, the decomposition
    $\Gamma$ comprises unitary operators which are implementable by local Pauli
    rotations interspersed with Pauli operators $-\ii \ H_k/\norm{H_k}$ for some
    term $k$. The double Hadamard test procedure from \Cref{sec:double_hadamard} is
    then implemented using controlled versions of these local circuits, as depicted
    in \Cref{fig:double_hadamard_test}, and the size of any such circuit is of the
    order $r\cdot \mathrm{poly}(n)$.

    To show the final claim, we give an
    explicit procedure for performing the sampling task. We have already defined $c_{k,0}$ and $c_{k,1}$ whenever $k\in\partial A$. Define also $c_{k,0}=1$ and $c_{k,1}=0$ for every $k\in E_A\cup E_B$, and let $E:=E_A\cup E_B\cup \partial A$. First, set $\ell=1$ and
    for each $k\in E$ independently sample the Bernoulli random variable
    $\bm{x}_{\ell,k}\in \{0,1\}$ where
    \begin{align}
        \bm{x}_{\ell,k} = \begin{cases}
                              0 & \text{w/ prob.}\ \ \frac{c_{k,0}}{c_{k,0}+c_{k,1}}  \\
                              1 & \text{w/ prob.}\ \ \frac{c_{k,1}}{c_{k,0}+c_{k,1}}.
                          \end{cases}
    \end{align}
    Then repeat this procedure for each $\ell \in \{2,\dots, r\}$. The end result is a random vector $\vec{\bm{x}}$ indexed by elements in $\{0,1\}^{[r]\times E}$. Repeat this entire procedure once more, resulting in a random vector $\vec{\bm{y}}\in \{0,1\}^{[r]\times E}$. We may interpret each fixed value $\vec{x}\in \{0,1\}^{[r]\times E}$ as an index set for the vectors $c$ in the decomposition $\Gamma$ with elements
    \begin{align}
        c_{\vec{x}}:=\prod_{\ell\in [r]}\prod_{k\in E}c_{k, x_{\ell, k}}.
    \end{align}
    The 1- and 2-norms of these vectors can then be written explicitly as
    \begin{align}
        \norm{c}_1 = \sum_{\vec{x}}c_{\vec{x}} = \left(\prod_{k\in E} (c_{k,0}+c_{k,1})\right)^r,\quad\norm{c}_2 = \sqrt{\sum_{\vec{x}}c_{\vec{x}}} =  \left(\prod_{k\in E}\left(c_{k,0}^2 + c_{k,1}^2\right)\right)^{r/2}.
    \end{align}
    Let $\bm{g}\in \{0,1\}$ be an independent Bernoulli random variable such that $\bm{g}=0$ with probability $1/2$, and let $\calB$ denote the event where $\vec{\bm{x}}=\vec{\bm{y}}$ and $\bm{g}=1$. To sample from the desired distribution, we post-select on $\calB$ \emph{not} occurring, i.e., the complement of $\calB$ which we denote by $\calB^c$. We can accomplish this post-selection by allowing one to repeat the procedure if $\calB$ occurs, up to $K$ times, and declaring failure if none of the trials yields the event $\calB^c$. Since $\Pr[\calB]\leq 1/2$ we can amplify the success probability to an arbitrarily small value by taking $K$ sufficiently large. We have in particular that
    \begin{align}
        \Pr[\calB^c] = 1-\frac{1}{2}\sum_{\vec{x}}\prod_{\ell\in [r]}\prod_{k\in E}\frac{c_{k,x_{\ell,k}}^2}{(c_{k,0}+c_{k,1})^2} & =1-\frac{1}{2}\frac{\norm{c}_2^2}{\norm{c}_1^2}.
    \end{align}
    Next, observe that the random variables $\vec{\bm{x}}$, $\vec{\bm{y}}$, and $\bm{g}$ are distributed such that for each value of $g\in \{0,1\}$ and $\vec{x},\vec{y}\in\{0,1\}^{[r]\times E}$ which are in the event $\calB^c$ we have
    \begin{align}
        \Pr[\vec{\bm{x}}=\vec{x},\ \vec{\bm{y}}=\vec{y},\ \bm{g}=g\ | \ \calB^c] & = \frac{1}{2}\prod_{\ell\in [r]}\prod_{k\in E}\frac{c_{k,x_{\ell,k}}}{c_{k,0}+c_{k,1}}\cdot \frac{c_{k,y_{\ell,k}}}{c_{k,0}+c_{k,1}}\left(1-\frac{1}{2}\frac{\norm{c}_2^2}{\norm{c}_1^2}\right)^{-1} \\
        & = \frac{c_{\vec{x}}c_{\vec{y}}}{2\norm{c}_1^2}\left(1-\frac{1}{2}\frac{\norm{c}_2^2}{\norm{c}_1^2}\right)^{-1}\\
                                                                                 & = \frac{c_{\vec{x}}c_{\vec{y}}}{2\norm{c}_1^2 - \norm{c}_2^2}                                                                                                                                                \\
                                                                                 & = c_{\vec{x}}c_{\vec{y}}\phi(\Gamma)^{-1}.
    \end{align}
    Thus, the setting random variables $\vec{\bm{x}}$, $\vec{\bm{y}}$, and $\bm{g}$ produced by this post-selected random process are distributed as in \cref{eq:algo_prob_dist}, as desired.
\end{proof}
We conclude with some additional observations about the circuits appearing in the procedure above which may be of interest. We state these without proof.
\begin{enumerate}
    \item For each possible circuit, the subgraph of the circuit interaction graph
          restricted to qubits in $A$ is identical to that for the Hamiltonian
          interaction graph, and similarly for $B$.
    \item For each possible circuit, the vertex corresponding to $R_A$ in the circuit
          interaction graph is adjacent only to those qubits in $A$ with which $\partial
              A$ is incident, and similarly for $R_B$.
\end{enumerate}
These observations are also depicted in \Cref{fig:hamiltonian_intro_fig}.
\section{Time-like cuts}\label{sec:time_like}
In this section, we analyze the performance of a specific time-like cut of the
identity channel (i.e., \Cref{def:time_like_cut} with $\calN_{A\to A} = \id_A$)
for a natural operational task. We pick a decomposition of the form in
\cref{eq:time_like_cut_defn} which is optimal, i.e., the 1-norm of the
time-like cut is equal to the time-like gamma factor
$\gamma_{\uparrow}(\id_A)=2d_A-1$. Additionally, the required
measure-and-prepare operations $\calM_i$ can be implemented efficiently using
diagonal 2-designs, and there is no post-processing of ancilla qubits required,
so $d_{R_B}=1$ and \cref{eq:time_like_cut_defn} becomes
$\id_A=\sum_{i}a_i\calM_i$. Note that the fact that the time-like gamma factor
is at most $2d_A-1$ is immediate from the decomposition we give, while the
argument for the lower bound is nearly identical to the proof of
\Cref{claim:1_norm_lb}, so we omit it here. See
\cite{Yuan2021universalmemories, brenner2023optimal} for more detailed
discussions. To analyze the performance of the time-like cut in an operational
task, we introduce the following template for an algorithm with desirable
properties.

% --------------------- time-like cutting pseudocode --------------------------------- %
\algrenewcommand\algorithmicrequire{\textbf{Input:}}
\algrenewcommand\algorithmicensure{\textbf{Output:}}
\floatname{algorithm}{Algorithm}
\begin{algorithm}[H]
    \caption{Mean estimation using time-like cut without ancillas}\label{proc:wire_cutting}
    \begin{algorithmic}[1]
        \Require $\rho_{AE}^{\otimes N}$, observable $X$ on $AE$, $\veps$, $d_A$
        \Ensure Estimate $\hat{\bm{\mu}}$ of $\Tr(X\rho_{AE})$
        \For{$k=1,\dots, N$}
        \State Sample $\bm{z}_k\sim p$
        \State Prepare $\bm{\rho}^{\prime}_k = (\calM_{\bm{z}_k}\otimes \id_E)(\rho)$ using $k^\text{th}$
        copy of $\rho$
        \State $\bm{x}_k\gets$ measure $X$ on $\bm{\rho}^{\prime}_k$
        \EndFor
        \State $\hat{\bm{\mu}}\gets \mathrm{ClassicalPostProcessing}((\bm{z}_1,\bm{x}_1),\dots,(\bm{z}_N,\bm{x}_N))$\label{line:classical_post_processing_wire}
        \State \Return $\hat{\bm{\mu}}$
    \end{algorithmic}
\end{algorithm}

To instantiate the algorithm, one specifies the classical post-processing step
and a choice of an ensemble of measure-and-prepare channels
$\{(p_z,\calM_z)\}$. We say that the algorithm is successful if its output
satisfies $|\hat{\bm{\mu}}-\Tr(X\rho_{AE})|\leq \veps$, and we would like to
bound the number of iterations, or \textit{copies}, $N$ required for the
algorithm to succeed with high probability using our time-like cut. Using the
reasoning based on Hoeffding's Inequality presented in \Cref{sec:qpd} and in
prior work, $N=O(\norm{a}_1^2/\veps^2)$ copies should suffice. In
\Cref{sec:good_mp_channels} we show that going beyond this analysis by bounding
the variance directly leads to an improved upper bound for some cases. In
\Cref{sec:lower_bound} we verify that this analysis is tight when $X$ is
rank-1, using an information-theoretic argument.

\subsection{The performance of optimal measure-and-prepare channels}\label{sec:good_mp_channels}
We show the following.
\begin{theorem}\label{thm:rank_dep_wire_cutting}
    Let $A$ be a subset of the qubits in an $n$-qubit quantum system. There exists a pair of
    measure-and-prepare channels
    $\calM_0,\calM_1: \sfL(\calH_A)\to\sfL(\calH_A)$ and a choice of
    ensemble distribution $p:\{0,1\}\to [0,1]$
    such that \Cref{proc:wire_cutting} succeeds with high probability
    using
    \begin{align}\label{eq:rank_dep_bound}
        N=O\left(d_A \veps^{-2}(1+\norm{\Tr_A(X^2)})\right)
    \end{align}
    copies of the unknown state. Furthermore, $\calM_0,\calM_1$ can be implemented
    using $O(\log^2(d_A))$ diagonal 2-qubit gates along with measurement and
    state preparation in the computational basis.
\end{theorem}

%\begin{proposition}\label{thm:rank_dep_wire_cutting_v1}
%    Let $AE$ be a quantum system. There exists a pair of
%    measure-and-prepare channels
%    $\calM_0,\calM_1\in \mathsf{C}(\calH_A)$, each implementable
%    using $O(\log^2(d_A))$ elementary gates, such that the
%    following holds. Let
%    $\calU_1, \calU_2\in \mathsf{C}(\calH_A\otimes\calH_E)$ be
%    a pair of unitary channels,
%    $O\in\mathsf{L}(\calH_A\otimes\calH_E)$ be a Hermitian
%    observable satisfying $\norm{O}_\infty\leq 1$, and
%    $\rho_{AE}$ be some initial state. Then the expected value
%    $\Tr(O (\calU_2\circ\calU_1)(\rho_{AE}))$ can be estimated to
%    within additive error $\veps$, except with probability $0.01$,
%    using
%    \begin{align}\label{eq:rank_dep_bound}
%        N=O\left(d_A \veps^{-2} (1+\norm{\Tr_A(\calU_2^*(O^2))}_\infty)\right)
%    \end{align}
%    rounds in which one randomly chooses $z=0,1$, applies
%    $\calU_2\circ (\calM_z\otimes \id_E)\circ\calU_1$ to
%    $\rho_{AE}$, and measures the resulting state in the
%    eigenbasis of $O$.
%\end{proposition}
Let us first remark on some consequences of the bound in
\cref{eq:rank_dep_bound}. The operator norm in the right-hand side of
\cref{eq:rank_dep_bound} can in turn be bounded by $d_A^{(q-1)/q} \norm{X^2}_q$
for any $q\geq 1$ using the results of~\cite[{Prop.~1}]{Rastegin_2012}. The
following two consequences are of particular interest. When $q=\infty$, we get
$N=O(d_A^2/\veps^2)$, which reproduces the results obtained in prior work.
(More precisely, the dependence on the dimension scales at most like
$(2d_A-1)^2$ in this case.) For constant error $\veps=O(1)$, taking $q=1$ and
using the fact that all the eigenvalues of $X$ have magnitude at most $1$, we
find $N=O(d_A r)$ where $r$ is the rank of the observable $X$. This implies,
for instance, that additive error estimates of the output probabilities of a
unitary quantum circuit can be computed using $O(d_A)$ rounds of the above
procedure, which is a quadratic improvement over the bounds in prior work and
enables us to conclude that the information-theoretic lower bound we derive in
\Cref{sec:lower_bound} is tight in some cases.

\subsubsection*{Proof of \Cref{thm:rank_dep_wire_cutting}}
The Choi state of the identity channel $\id_{A\to A}$ is just the maximally entangled state $\Phi_A$. We will describe a pair of efficiently implementable measure-and-prepare channels $\calM_0,\calM_1:\sfL(\calH_A)\to\sfL(\calH_A)$ satisfying
\begin{align}\label{eq:choi_constraint_isom}
    \Phi_A = d_A J_{\mathcal{M}_0} - (d_A-1) J_{\mathcal{M}_1}.
\end{align}
This is equivalent to showing $\id_A = d_A\calM_0 - (d_A-1)\calM_1$, which is a space-like cut of $\id_A$ with 1-norm $2d_A-1$. Let $\{\ket{1},\ket{2},\dots,\ket{d_A}\}$ denote the standard basis for $\calH_A$ and consider a uniformly random ``equatorial" state
\begin{align}
    \ket{v_{\bm{\theta}}} := \frac{1}{\sqrt{d_A}}\sum_{j=1}^{d_A} \mathrm{e}^{\mathrm{i}\bm{\theta}_j}\ket{j}
\end{align}
where $\bm{\theta}=\bm{\theta}_1\bm{\theta}_2\dots\bm{\theta}_{d_A}$ for $\bm{\theta}_j$ drawn independently and uniformly at random from $[0,2\pi)$. Define
\begin{align}
     & \calM_0\colon \rho \mapsto d_A\expct_{\bm{\theta}} \Tr(\outerprod{v_{\bm{\theta}}}{v_{\bm{\theta}}}\rho) \outerprod{v_{\bm{\theta}}}{v_\btheta},\label{eq:plus_channel} \\
     & \calM_1\colon \rho\mapsto \frac{1}{d_A-1}\sum_{k\neq \ell}\Tr(\outerprod{k}{k}\rho)\outerprod{\ell}{\ell}
\end{align}
where $k,\ell \in [d_A]$. Let us first check that this choice satisfies \cref{eq:choi_constraint_isom}. We have
\begin{align}
    J_{\calM_0} & =(\mathrm{id}\otimes \calM_0)(\Phi)                                                                                                           \\ &= \sum_{j,k=1}^{d_A}\outerprod{j}{k}\expct_{\bm{\theta}}\langle v_{\bm{\theta}}| j\rangle \langle k | v_{\bm{\theta}}\rangle \outerprod{v_{\bm{\theta}}}{v_{\bm{\theta}}}\\
                & = \expct\left[\outerprod{\overline{v_{\bm{\theta}}}}{\overline{v_{\bm{\theta}}}} \otimes \outerprod{v_{\bm{\theta}}}{v_{\bm{\theta}}}\right]
\end{align}
and
\begin{align}
    J_{\calM_1} = \frac{1}{d_A (d_A-1)}\sum_{k\neq \ell}\outerprod{k}{k}\otimes \outerprod{\ell}{\ell}.
\end{align}
On the other hand, it is fairly straightforward to verify that
\begin{align}\label{eq:2nd_moment_equatorial}
    \expct\left[(\outerprod{v_{\bm{\theta}}}{v_{\bm{\theta}}})^{\otimes 2}\right] = \frac{1}{d_A^2}\left(F+ \sum_{k\neq \ell} \outerprod{k\ell}{k\ell}\right)
\end{align}
where $F$ is the swap operation on $(\mathbb{C}^{d_A})^{\otimes 2}$.
By taking partial transposes of both sides, \cref{eq:2nd_moment_equatorial} holds if and only if
\begin{align}
    \expct\left[\outerprod{\overline{v_{\bm{\theta}}}}{\overline{v_{\bm{\theta}}}} \otimes \outerprod{v_{\bm{\theta}}}{v_{\bm{\theta}}}\right] & = \frac{1}{d_A}\left(\Phi+\frac{1}{d_A}\sum_{k\neq \ell} \outerprod{k\ell}{k\ell}\right)\label{eq:choi_constraint}                      \\
                                                                                                                                               & = \frac{1}{d_A}\left(\Phi+ \frac{d_A-1}{d_A(d_A-1)}\sum_{k\neq \ell} \outerprod{k\ell}{k\ell}\right)\label{eq:choi_state_of_plus_channel}
\end{align}
Hence, $\calM_0$ and $\calM_1$ satisfy \cref{eq:choi_constraint_isom}.

We now explain how they can be implemented efficiently on a quantum circuit on
$n$ qubits, assuming $d_A=2^n$. The channel $\calM_1$ is straightforward to
implement using measurements and state preparations in the computational basis,
so we focus on $\calM_0$. We claim that the following procedure implements the
channel $\calM_0$.
\begin{tcolorbox}[colback=white, arc=0pt, boxrule=0.5pt]
    \begin{protocol}[Optimal measure-and-prepare channel $\calM_0$]\label{prot:prot_1}$ $\vspace{0.5em}\newline
        \textbf{Input}: $\rho\in\mathsf{D}(\calH_A\otimes \calH_E)$ for $A$ an $n$-qubit system.\vspace{0.5em}\\
        \textbf{Output}: $(\calM_0\otimes \id_E)(\rho)$, where $\calM_0$ is as in \cref{eq:plus_channel}.
        \begin{enumerate}
            \item Apply a phase-random circuit $\bm{U}^\dag$ to the system $A$.
            \item Apply single-qubit Hadamard gates on each qubit in $A$.
            \item Measure the system $A$ in the computational basis, obtaining $x\in\{0,1\}^n$.
            \item Prepare the state $\bm{U} H^{\otimes n} \ket{x}_A$ on system $A$.
        \end{enumerate}
    \end{protocol}
\end{tcolorbox}
This implements a channel acting on subsystem $A$ with the action
\begin{align}\label{eq:new_eqn_label}
    \rho \mapsto \expct_{\bm{U}} \sum_{x\in\{0,1\}^n} \Tr(\bm{U}\outerprod{h_x}{h_x}\bm{U}^\dag \rho) \bm{U} \outerprod{h_x}{h_x}\bm{U}^\dag
\end{align}
for any $\rho\in \sfL(\calH_A)$ where we let $\ket{h_x}$ denote the state $H^{\otimes n} \ket{x}$. The right-hand side of the above is equal to
\begin{align}
    \sum_{x\in\{0,1\}^n}\expct_{\bm{U}} \Tr_1 \left((\rho\otimes \mathds{1}) (\bm{U}\outerprod{h_x}{h_x}\bm{U}^\dag)^{\otimes 2}\right) & = \sum_{x\in\{0,1\}^n}\Tr_1 \left((\rho\otimes \mathds{1}) \expct_{\bm{U}}\ (\bm{U}\outerprod{h_x}{h_x}\bm{U}^\dag)^{\otimes 2}\right)                                                     \\
                                                                                                                       & = \sum_{x\in\{0,1\}^n} \Tr_1 \left((\rho\otimes \mathds{1}) \expct_\btheta \ V_\btheta^{\otimes 2} \outerprod{h_x}{h_x}^{\otimes 2} (V_\btheta^{\dag})^{\otimes 2}\right) \\
                                                                                                                       & = d_A \Tr_1 \left((\rho\otimes \mathds{1}) \expct_\btheta\ (\outerprod{v_\btheta}{v_\btheta})^{\otimes 2}\right)                                                          \\
                                                                                                                       & = d_A\expct_{\bm{\theta}} \Tr(\outerprod{v_{\bm{\theta}}}{v_{\bm{\theta}}}\rho) \outerprod{v_{\bm{\theta}}}{v_\btheta}
\end{align}
where the second line follows since $\bm{U}$ forms a diagonal unitary $2$-design and the third line follows from the fact that $V_\btheta \ket{h_x}$ is identically distributed to $\ket{v_\btheta}$ for any $x\in\{0,1\}^n$. The total gate complexity of this procedure is $O(n^2)$ and it is dominated by the phase-random circuit.

It remains to bound the number of additional samples required to achieve the
simulation task using these measure-and-prepare channels. To this end, consider
randomly applying one of the two possible modified circuits in the above scheme
according to a Bernoulli random variable $\bm{z}$ such that $\bm{z}=0$ with
probability $d_A/(2d_A-1)$ and $\bm{z}=1$ with probability $(d_A-1)/(2d_A-1)$.
For each possible outcome $z=0,1$, this yields the final state
\begin{align}\label{eq:sigma_defn}
    \sigma_z := (\calM_{z}\otimes \id_E) (\rho_{AE}).
\end{align}
Next, fix an eigendecomposition of $X$ of the form
$X=\sum_{j=1}^{d_{AE}} \lambda_j \outerprod{v_j}{v_j}$,
such that $\lambda_1\geq \lambda_2\geq \dots \geq \lambda_{d_{AE}}$.
If we measure $X$ we obtain a random variable $\bm{y}\in \mathrm{spec}(X)$
such that, conditioned on $\bm{z}=z$, we have $\bm{y}=y$ with
probability
$\sum_{j\in [d_{AE}]:\lambda_j=y} \bra{v_j}\sigma_z\ket{v_j}$.
As described in \Cref{sec:qpd}, we take our
unbiased estimator of the true expectation
$\mu := \Tr(X\rho_{AE})$ to be
\begin{align}
    \hat{\bm{\mu}} := (2d_A-1)(-1)^{\bm{z}} \ \bm{y}.
\end{align}
\begin{claim}
    It holds that $\expct \hat{\bm{\mu}} = \mu$ and
    \begin{align}
        \Var[\hat{\bm{\mu}}]\leq (2 d_A - 1)\cdot \min\{2\norm{\Tr_A(X^2)}+1,\ 2d_A-1\}.
    \end{align}
\end{claim}
\begin{proof}
    First, we have
    \begin{align}
        \expct \hat{\bm{\mu}} & = d_A \expct [\bm{y}| \bm{z}=0] - (d_A-1)\expct [\bm{y}|\bm{z}=1] \\
                              & = d_A \Tr(X \sigma_0) - (d_A-1)\Tr(X \sigma_1)                    \\
                              & =\mu
    \end{align}
    where the last line follows by making use of
    \cref{eq:sigma_defn}, the linearity of trace,
    and the fact that
    $d_A\calM_0 - (d_A-1)\calM_1 = \id_{A}$. The
    bound $\Var[\hat{\bm{\mu}}]\leq (2d_A-1)^2$ follows from
    the definition of $\hat{\bm{\mu}}$, since
    $\norm{X}\leq 1$ and therefore
    $|\hat{\bm{\mu}}|\leq 2d_A-1$ with probability 1.
    Finally, we bound the variance by
    the second moment, which is equal to
    \begin{align}
        \expct \hat{\bm{\mu}}^2 & = (2d_A-1)^2\left(\frac{d_A}{2d_A-1}\expct \left[\bm{y}^2|\bm{z}=0\right] + \frac{d_A-1}{2d_A-1}\expct\left[\bm{y}^2|\bm{z}=1\right]\right) \\
                                & = (2d_A-1)\left(d_A\Tr(X^2 \sigma_0) + (d_A-1)\Tr(X^2 \sigma_1)\right).\label{eq:second_moment_bound}
    \end{align}
    Using the fact that $\calM_0$ and $\calM_1$ are self-adjoint maps, we have
    \begin{align}
        \Tr(X^2 \sigma_z) = \Tr((\calM_z\otimes\id_E)(X^2) \rho_{AE})
    \end{align}
    for each $z\in \{0,1\}$. We may then compute
    \begin{align}
        (\calM_0 \otimes \id_E) (X^2) & = d_A \expct_{\bm{\theta}} \Tr_2\left((\outerprod{v_{\bm{\theta}}}{v_{\bm{\theta}}}^{\otimes 2}_{12}\otimes \mathds{1}_3)(\mathds{1}_1\otimes X^2_{23})\right) \\
                                      & \preceq \frac{1}{d_A}\left[\Tr_2\left((F_{12}\otimes \mathds{1}_3)(\mathds{1}_1\otimes X^2_{23})\right)+\Tr_2\left(\mathds{1}_1\otimes X^2_{23}\right)\right]  \\
                                      & = \frac{1}{d_A}\left[X^2 + \mathds{1}_A \otimes  \Tr_{A}(X^2)\right]
    \end{align}
    where the second line follows from \cref{eq:2nd_moment_equatorial} as well as the fact that
    \begin{align}
        \sum_{k\neq \ell} \outerprod{k\ell}{k\ell} & = \mathds{1}\otimes\mathds{1} - \sum_{k} \outerprod{kk}{kk}
    \end{align}
    and the operator
    \begin{align}
        \Tr_2 \left((\outerprod{kk}{kk}_{12}\otimes\mathds{1}_3)(\mathds{1}_1\otimes X^2_{23})\right)
    \end{align}
    is positive semidefinite for all $k$.
    Therefore,
    \begin{align}
        \Tr(X^2\sigma_0) & \leq \frac{1}{d_A}\left[\Tr(X^2 \rho_{AE}) + \Tr(\Tr_A(X^2)\rho_E)\right]             \\
                         & \leq \frac{1}{d_A}\left(1 +\norm{\Tr_A(X^2)}\right).\label{eq:second_moment_helper_1}
    \end{align}
    where $\rho_E = \Tr_A(\rho_{AE})$ is the reduced state of the qubits which are not acted upon by the measure-and-prepare channels, and the second line follows since $\norm{O}\leq 1$. Similarly, we can bound the second term in \cref{eq:second_moment_bound} by observing that
    \begin{align}
        (\calM_1\otimes \id_E)(X^2) & = \frac{1}{d_A-1}\sum_{k\neq \ell} \Tr_2\left((\outerprod{k}{k}_1\otimes \outerprod{\ell}{\ell}_2\otimes \mathds{1}_3)(\mathds{1}_1\otimes X^2_{23})\right) \\
                                    & \preceq \frac{\Tr_2(\mathds{1}_1\otimes X^2_{23})}{d_A-1}                                                                                                   \\
                                    & = \frac{\mathds{1}_A\otimes \Tr_A(X^2)}{d_A-1}
    \end{align}
    and therefore
    \begin{align}
        \Tr(X^2\sigma_1) & \leq \frac{\norm{\Tr_A(X^2)}}{d_A-1}.\label{eq:second_moment_helper_2}
    \end{align}
    Substituting \cref{eq:second_moment_helper_1}
    and \cref{eq:second_moment_helper_2} into
    \cref{eq:second_moment_bound} yields the bound on the
    variance as claimed.
\end{proof}
Using the bound on the variance in the above claim as well as Chebyshev's Inequality concludes the proof of \Cref{thm:rank_dep_wire_cutting}. We remark that the result would follow from a similar analysis based on unitary 2-designs (e.g., random Clifford circuits) rather than diagonal 2-designs, though we chose the latter since the resulting pair of channels $\calM_0$, $\calM_1$ achieve the optimal space-like cut.

\subsection{An information-theoretic lower bound}\label{sec:lower_bound}
%Prop.~\ref{thm:rank_dep_wire_cutting} and the analysis in the previous section are reminiscent of the classical shadows framework of Ref.~\cite{Huang2020}. A key difference between the setting considered in that paper and the one here is that in the simulation task above we do not assume the ability to efficiently compute expectation values of the form $\Tr(O\outerprod{v}{v})$ for arbitrary vectors $\ket{v}$ in a $d$-dimensional Hilbert space. Instead, one is forced to compute such expectation values by repeatedly measuring a prepared state in the eigenbasis of $O$. Additionally, the procedure above works without needing to perform any randomized measurements on the subsystem $E$, which avoids an explicit scaling in $d_{AE}$ when compared with the worst-case scaling for classical shadows with random Clifford circuits. It may then be natural to ask whether one should expect the bound in Prop.~\ref{thm:rank_dep_wire_cutting} to be tight. Indeed, the naive analogue to the sample complexity obtained in classical shadows would be a bound which has no explicit dependence on the dimension $d_A$, but rather scales with a quantity that depends only on the observable of interest (as in the ``shadow-norm" of Ref.~\cite{Huang2020}.)

In this section, we show an information-theoretic lower bound of $\Omega(d_A)$
on the number of copies required in any instantiation of
\Cref{proc:wire_cutting} whenever $X$ is rank-1. In particular, this allows us
to conclude that the bound in \Cref{thm:rank_dep_wire_cutting} is tight for
the special case where one is interested in output probabilities of quantum
circuits. Furthermore, the lower bound suggests that, unlike in classical
shadows~\cite{Huang2020}, for example, a dependence on $d_A$ in the number of
samples required for circuit cutting is unavoidable even if one restricts the
task to estimating expectation values of low-rank observables.
\begin{theorem}\label{thm:time_like_lower_bound}
    Consider the setting in \Cref{proc:wire_cutting}.
    Suppose $X$ is known to be a rank-1 projection operator.
    Then the choice of measure-and-prepare channels
    in the proof of \Cref{thm:rank_dep_wire_cutting}
    is sample-optimal with
    respect to the dimension of subsystem $A$; that is,
    $N=\Theta(d_A)$ copies of the unknown state are both necessary and sufficient
    for a procedure of the form in \Cref{proc:wire_cutting}
    to succeed with high probability.
\end{theorem}
\begin{proof}
    The upper bound follows directly from \Cref{thm:rank_dep_wire_cutting} using the fact that
    $\norm{\Tr_A(X^2)}\leq 1$ in this case. For the lower bound,
    consider the state discrimination task in which the goal is to
    distinguish between the two alternatives
    $\rho^{(1)}_U=U\outerprod{1}{1}U^\dag$
    and $\rho^{(2)}_U=U\outerprod{2}{2}U^\dag$, where
    $U\in \sfU(\calH_A\otimes\calH_E)$ is some unitary operator.
    Clearly, this task
    reduces to estimating $\Tr(X \rho)$ for
    the unknown state $\rho\in \{\rho^{(1)}_U,\rho^{(2)}_U\}$ and
    with the observable taken to be $X=U\outerprod{1}{1}U^\dag$. Namely,
    if the estimate $\hat{\mu}$ is sufficiently accurate, then outputting
    $1$ if $\hat{\mu}\geq 1/2$ and $2$ otherwise results in a successful
    discrimination. Hence, it suffices to show the existence of a
    unitary $U$ such that the information available from the procedure
    described in \Cref{proc:wire_cutting} is insufficient
    to identify $\rho$ unless it is repeated $N=\Omega(d_A)$
    times. To this end, let $\bm{x}\in \{1,2\}$ be a random variable,
    let $\bm{U}$ be a Haar-random unitary, let $\bm{z}$ denote the choice
    of measure-and-prepare channel, and let
    $\bm{y}\in \{0,1\}$ be the random variable corresponding to measuring
    $(\calM_{\bm{z}}\otimes \id_E)(\rho^{(\bm{x})}_{\bm{U}})$ according to the
    POVM $\{M,\mathds{1}-M\}$ where $M=\bm{U}\outerprod{1}{1}\bm{U}^\dag$
    and $M$ is the POVM element corresponding to the outcome $\bm{y}=0$.
    Let us now define the following shorthand for the joint distribution
    of the random variables $(\bm{y},\bm{z})$ (which are the ones
    available for use in the state discrimination task) conditioned on
    the others. For each
    $x\in \{1,2\}$, $U\in\sfU(\calH_A\otimes\calH_B)$, and
    $z$ a possible value of $\bm{z}$ let
    \begin{align}
        p^{(x)}_U(0,z) & := \Pr[\bm{y}=0,\bm{z}=z| \bm{x}=x, \bm{U}=U]
    \end{align}
    and $p^{(x)}_U(1,z):=1-p^{(x)}_U(0,z)$. We claim that
    \begin{align}\label{eq:expected_tv_distance}
        \expct_{\bm{U}\sim \mathrm{Haar}}
        d_{\mathrm{TV}}(p^{(1)}_{\bm{U}}, p^{(2)}_{\bm{U}})
        \leq O\left(\frac{1}{d_A}\right).
    \end{align}
    Therefore there exists a fixed unitary $U\in\sfU(\calH_A\otimes\calH_E)$
    for which the TV distance between the distributions
    $p^{(1)}_U$ and $p^{(2)}_U$ is at most $O(1/d_A)$. By the telescoping
    property of the TV distance,
    \begin{align}
        d_{\mathrm{TV}}((p^{(1)}_{U})^{\otimes N}, (p^{(2)}_{U})^{\otimes N})
        \leq O\left(\frac{N}{d_A}\right)
    \end{align}
    and therefore the left-hand side is small unless $N=\Omega(d_A)$.
    It remains to show \cref{eq:expected_tv_distance}.
    Let $q$ denote the marginal distribution of $\bm{z}$ and
    define $p_{U,z}^{(x)}(y)=p_{U}^{(x)}(y,z)/q(z)$. Then the left-hand
    side of \cref{eq:expected_tv_distance} is equal to
    $
        \expct_{\bm{z}\sim q} \expct_{\bm{U}\sim \mathrm{Haar}}
        d_{\mathrm{TV}}(p^{(1)}_{\bm{U},\bm{z}}, p^{(2)}_{\bm{U},\bm{z}})
    $
    since $\bm{U}$ and $\bm{z}$ are independent. Also, for any $z$ we have
    \begin{align}
        \expct_{\bm{U}\sim\mathrm{Haar}}
        d_{\mathrm{TV}}(p_{\bm{U},z}^{(1)},p_{\bm{U},z}^{(2)})
         & = \expct_{\bm{U}\sim\mathrm{Haar}}
        \left\lvert p_{\bm{U},z}^{(1)}(0) - p_{\bm{U},z}^{(2)}(0)\right\rvert \\
         & \leq \expct_{\bm{U}\sim\mathrm{Haar}}
        \left(p_{\bm{U},z}^{(1)}(0) + p_{\bm{U},z}^{(2)}(0)\right).\label{eq:triangle_ineq_tv}
    \end{align}
    Hence, it suffices to show that both terms in \cref{eq:triangle_ineq_tv} are $O(1/d_A)$ for any value of $z$.
    Let $\{E_j\}\subset \sfL(\calH_A\otimes\calH_E)$ be the
    Kraus operators corresponding to the
    channel $\calM_z\otimes \id_E$. We let $F_{\pi}$ denote the permutation operator corresponding to the permutation $\pi\in \mathfrak{S}_4$. For the first term, we compute
    \begin{align}
        \expct_{\bm{U}\sim\mathrm{Haar}}p_{\bm{U},z}^{(1)}(0)
         & = \sum_j\expct_{\bm{U}}
        \Tr\left\{\bm{U}\outerprod{1}{1}\bm{U}^\dag E_j \bm{U}\outerprod{1}{1}\bm{U}^\dag E_j^\dag\right\}                                                             \\
         & = \sum_j \Tr\left\{F_{(1324)}\left(\expct \ (\bm {U}\outerprod{1}{1}\bm{U}^\dag)^{\otimes 2}\otimes E_j\otimes E_j^\dag\right)\right\}.\label{eq:first_prob}
    \end{align}
    Using the well-known identity
    \begin{align}
        \expct_{\bm{\varphi}\sim\mathrm{Haar}} \outerprod{\bm{\varphi}}{\bm{\varphi}}
         & = \frac{1}{d(d+1)}\left(\mathds{1}\otimes\mathds{1} + F_{(12)}\right)
    \end{align}
    in dimension $d$, we can rewrite the $j^\text{th}$ term in the right-hand side of \cref{eq:first_prob} as
    \begin{align}
         & \frac{1}{d(d+1)}\left[\vphantom{\Tr\left\{E_j^\dag\right\}}\right.
            \underbrace{\Tr\left\{F_{(1324)}(\mathds{1}^{\otimes 2}\otimes E_j\otimes E_j^\dag)\right\}}_{=\Tr(E_j^\dag E_j)}
            + \underbrace{\Tr\left\{F_{(13)(24)}(\mathds{1}^{\otimes 2}\otimes E_j\otimes E_j^\dag)\right\}}_{=|\Tr(E_j)|^2}\left. \vphantom{\Tr\left\{E_j^\dag\right\}}\right]
    \end{align}
    where in the above and until the end of this proof we are setting $d:=d_Ad_E$.
    Therefore, we have
    \begin{align}
        \expct_{\bm{U}\sim\mathrm{Haar}}p_{\bm{U},z}^{(1)}(0)
         & = \frac{1}{d(d+1)}\left(\sum_j \Tr(E_j^\dag E_j) + \sum_j |\Tr(E_j)|^2\right) \\
         & = \frac{1}{d+1}+ \frac{\sum_j|\Tr(E_j)|^2}{d(d+1)}                            \\
         & \leq \frac{1}{d+1}+ \frac{d_E^2d_A}{d(d+1)}                                   \\
         & =O\left(\frac{1}{d_A}\right).
    \end{align}
    Here, the second line uses the fact that the Kraus operators satisfy
    $\sum_jE_j^\dag E_j = \mathds{1}_{AE}$. The third line is based on the
    following reasoning. Since
    the measure-and-prepare channels act trivially on the $E$ subsystem
    we may write $E_j=\outerprod{\psi_j}{\phi_j}\otimes \mathds{1}_E$ for some
    normalized $\ket{\psi_j}$ and potentially unnormalized
    $\ket{\phi_j}$ satisfying $\norm{\ket{\phi_j}}\leq 1$ and
    $\sum_j\outerprod{\phi_j}{\phi_j} = \mathds{1}_A$. (Using rank-1 Kraus
    operators for measure-and-prepare channels are without loss of generality
    by~\cite[{Thm.~4}]{Horodecki2003entanglementbreaking}.) Thus,
    \begin{align}
        \sum_j|\Tr(E_j)|^2 = d_E^2 \sum_j |\langle \phi_j | \psi_j \rangle|^2\leq d_E^2 \sum_j \langle \phi_j | \phi_j\rangle = d_E^2 d_A.
    \end{align}
    Finally, we apply a similar argument to bound the second term in \cref{eq:triangle_ineq_tv}.
    We have
    \begin{align}
        \expct_{\bm{U}\sim\mathrm{Haar}}p^{(2)}_{\bm{U},z} & = \sum_j \Tr\left\{F_{(1324)}\left(\expct \ \bm{U}^{\otimes 2}(\outerprod{1}{1}\otimes \outerprod{2}{2})(\bm{U}^\dag)^{\otimes 2}\otimes E_j\otimes E_j^\dag\right)\right\} \\
                                                           & = \frac{1}{d^2-1}\sum_j \left[\Tr(E_j^\dag E_j) - \frac{|\Tr(E_j)|^2}{d}\right]                                                                                             \\
                                                           & \leq \frac{d}{d^2-1}                                                                                                                                                        \\
                                                           & = O\left(\frac{1}{d_A}\right)
    \end{align}
    where the second line follows from the identity
    \begin{align}
        \expct_{\bm{U}\sim\mathrm{Haar}} \bm{U}^{\otimes 2}(\outerprod{u}{u}\otimes \outerprod{v}{v})(\bm{U}^\dag)^{\otimes 2}
         & = \frac{1}{d^2-1}\left(\mathds{1}\otimes \mathds{1} - \frac{F_{(12)}}{d}\right)
    \end{align}
    for any two orthogonal unit vectors $\ket{u},\ket{v}\in\mathbb{C}^d$
    and the third line follows from neglecting the second term and once again
    noting that $\sum_j E_j^\dag E_j = \mathds{1}_{AE}$.
\end{proof}

\section{Further directions}
Our work raises several open questions. Firstly, does there exist a bipartite
unitary $U$ for which $\xi(U)\neq R_c(U)$? The procedure for space-like cutting
described here (and also in simultaneous work~\cite{schmitt2023cutting}) shows
that classical communication does not lead to a lower 1-norm in a space-like
cut for a large class of unitaries. Is there an entangling operation for which
classical communication provably lowers the minimal 1-norm in a space-like cut,
as originally suggested in Ref.~\cite{piveteau2023circuit}?

It is also natural to ask how far techniques for circuit cutting can be pushed
from an information-theoretic standpoint. Can one show that any choice of
measure-and-prepare channel (and post-processing function) in
\Cref{proc:wire_cutting} necessarily incurs a sample overhead of $\Omega(4^k)$
for general observables, matching the upper bound? Note that this is false if
we relax \Cref{proc:wire_cutting} to allow access to the intermediate
measurement outcomes obtained during application of the measure-and-prepare
channels. In this setting, when the register $E$ is trivial (one ``cuts" all
the wires), the observable outcomes may be disregarded completely, and one may
perform classical shadows~\cite{Huang2020} on the wires to predict the
expectation value using at most $O(2^k)$ samples of the unknown state, though
perhaps computationally inefficiently. Could the answer depend on assumptions
regarding computational efficiency? What should one expect of an
information-theoretic lower bound for space-like cutting?

It would be interesting and potentially useful to extend the ``double Hadamard
test" construction to general multipartite systems, and apply this to clustered
Hamiltonian simulation as well. Another direction would be to investigate the
possibility of computing spatial
correlation functions in thermal states or ground states using fewer qubits than might be expected. It would also be
interesting to see how circuit cutting techniques may be applied to compute
temporal correlation functions. For example, consider a correlation function of
the form $C_{PQ}(t):=\bra{\psi_0}\ee^{-\ii H t} P \ee^{\ii H t} Q \ket{\psi_0}$
where $\ket{\psi_0}$ is some initial tensor product state, $P$ and $Q$ are two
multi-qubit Pauli operators, and $H$ has interaction strength $\eta$ across
some partition. We may then estimate the magnitude $|C_{PQ}(t)|$ using local
circuits of the form used in \Cref{thm:clustered_ham_sim} through a Trotter
decomposition of $\ee^{-\ii H t} P \ee^{\ii H t} Q$ and taking the observable
to be $\outerprod{\psi_0}{\psi_0}$. The cost would then be on the order of
$\ee^{O(\eta t)}/\veps^4$. Is there a way to estimate this quantity using
similar techniques, including the sign? What are some specific examples of
quantum systems which are amenable to techniques for clustered Hamiltonian
simulation?

Finally, we conclude by reiterating an open question raised in
Ref.~\cite{bravyi2023classical} regarding the power of limited quantum memory.
Can one provably simulate a restricted, yet classically-hard family of
$n$-qubit quantum circuits (e.g., shallow circuits) in time $\text{poly}(n)$
using \textit{far} fewer qubits than expected, for example, $O(\text{poly}(\log
    n))$? As remarked by the authors of~\cite{bravyi2023classical}, such a
simulation might be enabled by the techniques considered in their work. Note
that naively applying the circuit cutting methods discussed in this work, one
could only hope to reduce the number of qubits required for such a simulation
by a constant factor, generically. We view this as an exciting direction of
both theoretical and practical importance.

\section{Acknowledgments}
AL thanks Richard Allen for helpful comments and Christophe Piveteau for discussions.
This work is supported by a collaboration between the US DOE and other Agencies. This material is based upon work supported by the U.S. Department of Energy, Office of Science, National Quantum Information Science Research Centers, Quantum Systems Accelerator. Additional support is acknowledged from the NSF for AL (grant PHY-1818914) and  AWH (CCF-1729369).

\appendix
\section{Robustness of the Choi state lower-bounds 1-norm}\label{sec:choi_robustness_bounds_1_norm_proof}
In this section we prove \Cref{claim:1_norm_lb}. The proof uses the same idea
as that in~\cite[Lemma 3.1]{piveteau2023circuit}, though we need to generalize
it slightly to the definition of a space-like cut presented here,
\Cref{def:space_like_cut}.
\begin{lemma}\label{lem:lower_bound_sep_decomp}
    Let $\rho\in\sfD(\calH_{AB})$ be a bipartite quantum state. Suppose there exist separable states $\sigma_1,\sigma_2,\dots\in\mathsf{SEP}(\calH_{AB}|A,B)$ and coefficients $a_1,a_2,\dots\in \mathbb{R}$ such that
    \begin{align}\label{eq:sep_state_decomp}
        \rho = \sum_j a_j \sigma_j.
    \end{align}
    It holds that $\sum_j|a_j|\geq 1+2R(\rho)$.
\end{lemma}
\begin{proof}
    We may rewrite \Cref{eq:sep_state_decomp} as
    \begin{align}
        \rho & = \sum_{j:a_j\geq 0} |a_j| \sigma_j - \sum_{j:a_j< 0} |a_j|\sigma_j                                                    \\
             & = \kappa_+ \sum_{j:a_j\geq 0} \frac{|a_j|}{\kappa_+} \sigma_j - \kappa_-\sum_{j:a_j< 0}\frac{|a_j|}{\kappa_-} \sigma_j \\
             & = (1+\kappa_-)\sigma_+ - \kappa_- \sigma_-\label{eq:rewritten_separable_decomp}
    \end{align}
    where in the second line we defined $\kappa_+ = \sum_{j:a_j\geq 0}|a_j|$ and $\kappa_- = \sum_{j:a_j< 0}|a_j|$, and $\sigma_+,\sigma_-$ are separable states, and the third line follows from the observation that $1=\Tr(\rho) = \sum_ja_j = \kappa_+-\kappa_-$. Comparing \Cref{eq:rewritten_separable_decomp} to the definition of robustness in \Cref{eq:pure_robustness_defn}, we necessarily have that
    \begin{align}
        R(\rho)\leq \kappa_- & = \frac{\kappa_+ + \kappa_- -\kappa_+ + \kappa_-}{2} = \frac{\sum_j|a_j| - 1}{2}.\nonumber\qedhere
    \end{align}
\end{proof}
In the remainder of the proof we give distinct labels to the input and output
systems of the channel we consider. Let $\calN:
    \sfL(\calH_{A_1B_1}) \to\sfL(\calH_{A_2B_2})$ be a bipartite quantum channel which
has a QPD of the form in \Cref{eq:qpd_defn} into \textit{separable} channels,
i.e.,
$
    \calN = \sum_j c_j \calT_j\circ \calE_j
$
for some $c_j\in\mathbb{R}$ satisfying $\sum_j|c_j|=\kappa$, separable channels $\calE_j:\sfL(\calH_{A_1B_1})\to\sfL(\calH_{A_2R_A}\otimes\calH_{B_2R_B})$, and post-processing functions $\calT_j:\sfL(\calH_{A_2R_A}\otimes\calH_{B_2R_B})\to\sfL(\calH_{A_2B_2})$ with the actions
\begin{align}
    \calT_j\colon \rho_{A_2R_AB_2R_B}\mapsto \Tr_{R_A R_B}((O_j\otimes \mathds{1}_{A_2B_2})\rho_{A_2R_AB_2R_B})
\end{align}
for some $O_j$ of the form $O_j=O^{(A)}_j\otimes O^{(B)}_j$ such that $\norm{O_j}\leq 1$. For each $j$, we have that $J_{\calE_j}\in \mathsf{SEP}(\calH_{A_1A_2R_AB_1B_2R_B}| A_1A_2R_A,B_1B_2R_B)$ by definition, so we may write
\begin{align}
    J_{\calE_j} = \sum_k p^{(j)}(k)\ \rho^{(j)}_k\otimes \sigma^{(j)}_k
\end{align}
where $p^{(j)}(k) > 0$, $\sum_k p^{(j)}(k) = 1$, $\rho^{(j)}_k\in \sfD(\calH_{A_1A_2R_A})$, and $\sigma^{(j)}_k\in\sfD(\calH_{B_1B_2R_B})$.
Then the Choi state $J_\calN\in \sfD(\calH_{A_1B_1}\otimes \calH_{A_2B_2})$ is equal to
\begin{align}
    J_{\calN} & = \sum_{jk}c_jp^{(j)}(k) (\id_{A_1 B_1}\otimes \calT_j)\left(\rho^{(j)}_k\otimes \sigma^{(j)}_k\right)                                                                                       \\
              & = \sum_{jk} c_jp^{(j)}(k) \Tr_{R_A}\left((\mathds{1}_{A_1A_2 }\otimes O^{(A)}_j)\rho_{k}^{(j)} \right)\otimes \Tr_{R_B}\left((\mathds{1}_{B_1B_2 }\otimes O^{(B)}_j)\sigma_{k}^{(j)} \right) \\
              & = \sum_{jk} c_jp^{(j)}(k) \sum_{x\in [d_A]}\sum_{y\in [d_B]}g_j(x,y) p^{(j)}_{A,k}(x)p^{(j)}_{B,k}(y)\omega^{(j,x)}_{A_1A_2,k}\otimes \tau^{(j,y)}_{B_1 B_2,k}
\end{align}
where in the third line we let $\{\ket{j,x}\}_x$ and $\{\ket{j,y}\}_y$ be eigenbases for $O^{(A)}_j$ and $O^{(B)}_j$, respectively, we let $g_j(x,y)$ be the $(x,y)^{\text{th}}$ eigenvalue of $O_j$, and we define
\begin{align}
    p^{(j)}_{A,k}(x) := \Tr\left((\mathds{1}_{A_1A_2}\otimes \outerprod{j,x}{j,x})\rho^{(j)}_{k}\right),\qquad p^{(j)}_{B,k}(x) := \Tr\left((\mathds{1}_{B_1B_2}\otimes \outerprod{j,y}{j,y})\sigma^{(j)}_k\right)
\end{align}
and
\begin{align}
    \omega_{A_1A_2,k}^{(j,x)} := \Tr_{R_A}\left((\mathds{1}_{A_1A_2}\otimes \outerprod{j,x}{j,x})\rho^{(j)}_{k}\right)/p_{A,k}^{(j)}(x),\qquad \tau_{B_1B_2,k}^{(j,y)} = \Tr_{R_B}\left((\mathds{1}_{B_1B_2}\otimes \outerprod{j,y}{j,y})\rho^{(j)}_{k}\right)/p_{B,k}^{(j)}(y).
\end{align}
By \Cref{lem:lower_bound_sep_decomp} we have
\begin{align}
    1+2R(J_{\calU}) & \leq \sum_{j}\sum_k\sum_{x\in [d_A]}\sum_{y\in [d_B]} |c_jg_j(x,y)p^{(j)}(k)p^{(j)}_{A,k}(x)p^{(j)}_{B,k}(y)|      \\
                    & \leq  \sum_j|c_j|\sum_k\sum_{x\in [d_A]}\sum_{y\in [d_B]}p^{(j)}(k) p^{(j)}_{A,k}(x)p^{(j)}_{B,k}(y) = \sum_j|c_j|
\end{align}
using the fact that $|g_j(x,y)|\leq 1$ and $\sum_{x\in [d_A]}p^{(j)}_{A,k}(x) = \sum_{y\in [d_B]}p^{(j)}_{B,k}(y)=1$ for any $j$ and $k$. This establishes the lower bound on the QPD 1-norm of $\calN$ in terms of the robustness of its Choi state.

\section{The product extent is well-defined}\label{sec:prod_extent_well_defined}
In this section we prove that the product extent (\Cref{defn:product_extent})
is well-defined, using elementary facts from linear programming. (See
Ref.~\cite[Chapter 4]{matousek2007understanding} for an introduction to the
relevant concepts.) Let $U\in\sfU(\calH_{AB})$ be a bipartite unitary operator.
For any positive integer $m\geq d_A^2d_B^2$ define $\xi_m(U)$ by a restriction
of the optimization problem in \Cref{defn:product_extent} to column vectors
with $m$ entries through
\begin{align}
    \begin{aligned}\label{eq:opt_problem_m_prod_extent}
        \xi_m(U):=\min\quad & 2\norm{c}_1^2 - \norm{c}_2^2        \\
        \textrm{s.t.} \quad & \sum_{j=1}^mc_j V_j\otimes W_j = U  \\
                            & c\in \mathbb{R}^m                   \\
                            & (V_j)_{j=1}^m\subset \sfU(\calH_A)  \\
                            & (W_j)_{j=1}^m\subset \sfU(\calH_B).
    \end{aligned}
\end{align}
That this quantity is well-defined follows from the fact that the objective function
is continuous and the feasible set defined by the constraints is nonempty (decompose $U$ in the Pauli basis) and compact. Clearly, we have $\xi_n(U)\leq \xi_{m}(U)$ for all $m,n\in \mathbb{Z}$ such that $d_A^2d_B^2\leq m\leq n$. Also, from the definition of the product extent and the fact that $\xi(U)\geq 1$ (\Cref{lem:extent_robustness_bounds}) we have $\xi(U) = \lim_{m\to \infty}\xi_m(U)$. It therefore suffices to show there exists some positive $m^*\in \mathbb{Z}$ such that for all $m\geq m^*$ we have $\xi_m(U)\geq \xi_{m^*}(U)$ since this implies that $\xi(U)=\xi_{m^*}(U)$ and the minimum in \Cref{defn:product_extent} is attained. To this end, let $m^* = 2d_A^2d_B^2$ and consider $\xi_m(U)$ for some $m \geq m^*+1$. Let $c\in\mathbb{R}^m$, $c\geq 0$ and $(V_j)_{j=1}^m$, $(W_j)_{j=1}^m$ be an optimal solution to the optimization problem in \Cref{eq:opt_problem_m_prod_extent}. (We may take $c\geq 0$ without loss of generality since the sign can be absorbed into the
unitary operators in the first constraint without changing the value of the objective function.) For each $\gamma\in \mathbb{R}$ define
\begin{align}
    S(\gamma):=\{d\in \mathbb{R}^m: d\geq 0,\ U=\sum_{j=1}^m d_j V_j\otimes W_j,\ \norm{d}_1=\gamma\}.
\end{align}
Then $S(\norm{c}_1)$ is a nonempty, convex, compact set. Hence, the convex optimization $\max\{\norm{d}_2: d\in S(\norm{c}_1)\}$ attains its maximum at an extreme point of $S(\norm{c}_1)$. But $S(\norm{c}_1)$ is a polytope specified by the $2d_A^2d_B^2$ linear constraints given by the real and imaginary parts of the equation $U=\sum_{j=1}^m d_j V_j\otimes W_j$. This implies that the extreme points have support of size at most $2d_A^2d_B^2$ by the equivalence between extreme points and basic feasible solutions for convex polytopes. Letting $d^*$ denote such an optimal extreme point, we therefore have
\begin{align}
    \xi_{m^*}(U)\leq 2\norm{d^*}_1^2 - \norm{d^*}_2^2 & = 2\norm{c}_1^2 - \norm{d^*}_2^2\leq 2\norm{c}_1^2 - \norm{c}_2^2 = \xi_m(U).
\end{align}
\printbibliography

@article{yuan2021hybridtn,
  title = {Quantum Simulation with Hybrid Tensor Networks},
  author = {Yuan, Xiao and Sun, Jinzhao and Liu, Junyu and Zhao, Qi and Zhou, You},
  journal = {Phys. Rev. Lett.},
  volume = {127},
  issue = {4},
  pages = {040501},
  numpages = {6},
  year = {2021},
  month = {Jul},
  publisher = {American Physical Society},
  doi = {10.1103/PhysRevLett.127.040501},
  url = {https://link.aps.org/doi/10.1103/PhysRevLett.127.040501}
}

@article{sun2022perturbative,
  title = {Perturbative Quantum Simulation},
  author = {Sun, Jinzhao and Endo, Suguru and Lin, Huiping and Hayden, Patrick and Vedral, Vlatko and Yuan, Xiao},
  journal = {Phys. Rev. Lett.},
  volume = {129},
  issue = {12},
  pages = {120505},
  numpages = {7},
  year = {2022},
  month = {Sep},
  publisher = {American Physical Society},
  doi = {10.1103/PhysRevLett.129.120505},
  url = {https://link.aps.org/doi/10.1103/PhysRevLett.129.120505}
}

@book{matousek2007understanding,
  ISBN = {9783540306979},
  url = {http://dx.doi.org/10.1007/978-3-540-30717-4},
  DOI = {10.1007/978-3-540-30717-4},
  publisher = {Springer Berlin Heidelberg},
  year = {2007},
  title={Understanding and Using Linear Programming},
  author={Ji\v{r}\'{\i} Matou\v{s}ek and Bernd G\"{a}rtner}
}

@article{plenio2005lognegativity,
  title = {Logarithmic Negativity: A Full Entanglement Monotone That is not Convex},
  author = {Plenio, M. B.},
  journal = {Phys. Rev. Lett.},
  volume = {95},
  number = {9},
  pages = {090503},
  numpages = {4},
  year = {2005},
  publisher = {American Physical Society},
  doi = {10.1103/PhysRevLett.95.090503},
  url = {https://link.aps.org/doi/10.1103/PhysRevLett.95.090503},
  eprint={quant-ph/0505071},
  archivePrefix={arXiv}
}

@misc{bravyi2023classical,
      title={Classical simulation of peaked shallow quantum circuits}, 
      author={Sergey Bravyi and David Gosset and Yinchen Liu},
      year={2023},
      eprint={2309.08405},
      archivePrefix={arXiv}
}

@article{Ortiz_2001,
   title={Quantum algorithms for fermionic simulations},
   volume={64},
   ISSN={1094-1622},
   url={http://dx.doi.org/10.1103/PhysRevA.64.022319},
   DOI={10.1103/physreva.64.022319},
   number={2},
   journal={Physical Review A},
   publisher={American Physical Society (APS)},
   author={Ortiz, G. and Gubernatis, J. E. and Knill, E. and Laflamme, R.},
   year={2001},
   eprint={cond-mat/0012334},
   archivePrefix={arXiv}
}

@article{Bravyi2019simulationofquantum,
  doi = {10.22331/q-2019-09-02-181},
  url = {https://doi.org/10.22331/q-2019-09-02-181},
  title = {Simulation of quantum circuits by low-rank stabilizer decompositions},
  author = {Bravyi, Sergey and Browne, Dan and Calpin, Padraic and Campbell, Earl and Gosset, David and Howard, Mark},
  journal = {{Quantum}},
  issn = {2521-327X},
  publisher = {{Verein zur F{\"{o}}rderung des Open Access Publizierens in den Quantenwissenschaften}},
  volume = {3},
  pages = {181},
  year = {2019},
  eprint={1808.00128},
  archivePrefix={arxiv}
}

@misc{ufrecht2023optimal,
      title={Optimal joint cutting of two-qubit rotation gates}, 
      author={Christian Ufrecht and Laura S. Herzog and Daniel D. Scherer and Maniraman Periyasamy and Sebastian Rietsch and Axel Plinge and Christopher Mutschler},
      year={2023},
      eprint={2312.09679},
      archivePrefix={arXiv}
}

@misc{marshall2023qubit,
      title={All this for one qubit? Bounds on local circuit cutting schemes}, 
      author={Simon C. Marshall and Jordi Tura and Vedran Dunjko},
      year={2023},
      eprint={2303.13422},
      archivePrefix={arXiv}
}

@article{Matthews2009distinguishability,
  title = {Distinguishability of Quantum States Under Restricted Families of Measurements with an Application to Quantum Data Hiding},
  volume = {291},
  ISSN = {1432-0916},
  url = {http://dx.doi.org/10.1007/s00220-009-0890-5},
  DOI = {10.1007/s00220-009-0890-5},
  number = {3},
  journal = {Communications in Mathematical Physics},
  publisher = {Springer Science and Business Media LLC},
  author = {Matthews,  William and Wehner,  Stephanie and Winter,  Andreas},
  year = {2009},
  pages = {813–843},
  eprint={0810.2327},
  archivePrefix={arXiv}
}

@ARTICLE{divincenzo2002datahiding,

  author={DiVincenzo, D.P. and Leung, D.W. and Terhal, B.M.},

  journal={IEEE Transactions on Information Theory}, 

  title={Quantum data hiding}, 

  year={2002},

  volume={48},

  number={3},

  pages={580-598},

  doi={10.1109/18.985948},

  eprint={quant-ph/0103098},
  archivePrefix={arXiv}
}

@misc{zhao2023power,
      title={Power of quantum measurement in simulating unphysical operations}, 
      author={Xuanqiang Zhao and Lei Zhang and Benchi Zhao and Xin Wang},
      year={2023},
      eprint={2309.09963},
      archivePrefix={arXiv}
}

@misc{schmitt2023cutting,
      title={Cutting circuits with multiple two-qubit unitaries}, 
      author={Lukas Schmitt and Christophe Piveteau and David Sutter},
      year={2023},
      eprint={2312.11638},
      archivePrefix={arXiv}
}

@article{Rastegin_2012,
   title={Relations for Certain Symmetric Norms and Anti-norms Before and After Partial Trace},
   volume={148},
   ISSN={1572-9613},
   url={http://dx.doi.org/10.1007/s10955-012-0569-8},
   DOI={10.1007/s10955-012-0569-8},
   number={6},
   journal={Journal of Statistical Physics},
   publisher={Springer Science and Business Media LLC},
   author={Rastegin, Alexey E.},
   year={2012},
   pages={1040–1053},
   eprint={1202.3853},
   archivePrefix={arXiv}
}

@article{Huang2020,
  title = {Predicting many properties of a quantum system from very few measurements},
  volume = {16},
  ISSN = {1745-2481},
  url = {http://dx.doi.org/10.1038/s41567-020-0932-7},
  DOI = {10.1038/s41567-020-0932-7},
  number = {10},
  journal = {Nature Physics},
  publisher = {Springer Science and Business Media LLC},
  author = {Huang,  Hsin-Yuan and Kueng,  Richard and Preskill,  John},
  year = {2020},
  pages = {1050–1057},
  eprint={2002.08953},
  archivePrefix={arXiv}
}

@article{eddins2022doubling,
  title = {Doubling the Size of Quantum Simulators by Entanglement Forging},
  author = {Eddins, Andrew and Motta, Mario and Gujarati, Tanvi P. and Bravyi, Sergey and Mezzacapo, Antonio and Hadfield, Charles and Sheldon, Sarah},
  journal = {PRX Quantum},
  volume = {3},
  number = {1},
  pages = {010309},
  numpages = {15},
  year = {2022},
  month = {Jan},
  publisher = {American Physical Society},
  doi = {10.1103/PRXQuantum.3.010309},
  url = {https://link.aps.org/doi/10.1103/PRXQuantum.3.010309},
  eprint={2104.10220},
  archivePrefix={arXiv}
}

@misc{pednault2023alternative,
      title={An alternative approach to optimal wire cutting without ancilla qubits}, 
      author={Edwin Pednault},
      year={2023},
      eprint={2303.08287},
      archivePrefix={arXiv}
}

@misc{harada2023optimal,
      title={Optimal parallel wire cutting without ancilla qubits}, 
      author={Hiroyuki Harada and Kaito Wada and Naoki Yamamoto},
      year={2023},
      eprint={2303.07340},
      archivePrefix={arXiv}
}

@misc{brenner2023optimal,
      title={Optimal wire cutting with classical communication}, 
      author={Lukas Brenner and Christophe Piveteau and David Sutter},
      year={2023},
      eprint={2302.03366},
      archivePrefix={arXiv}
}

@article{Childs_2021_trotter_error,
	doi = {10.1103/physrevx.11.011020},
  
	url = {https://doi.org/10.1103%2Fphysrevx.11.011020},
  
	year = 2021,
	month = {feb},
  
	publisher = {American Physical Society ({APS})},
  
	volume = {11},
  
	number = {1},
  
	author = {Andrew M. Childs and Yuan Su and Minh C. Tran and Nathan Wiebe and Shuchen Zhu},
  
	title = {Theory of Trotter Error with Commutator Scaling},
  
	journal = {Physical Review X},
    eprint={1912.08854},
    archivePrefix={arXiv}
}

@article{Nakata_2014,
	doi = {10.1088/1367-2630/16/5/053043},
  
	url = {https://doi.org/10.1088%2F1367-2630%2F16%2F5%2F053043},
  
	year = 2014,
	month = {may},
  
	publisher = {{IOP} Publishing},
  
	volume = {16},
  
	number = {5},
  
	pages = {053043},
  
	author = {Yoshifumi Nakata and Masato Koashi and Mio Murao},
  
	title = {Generating a state $t$-design by diagonal quantum circuits},
  
	journal = {New Journal of Physics},
    eprint={1311.1128},
    archivePrefix={arXiv}
}

@article{Mitarai2021constructing,
  doi = {10.1088/1367-2630/abd7bc},
  url = {https://doi.org/10.1088/1367-2630/abd7bc},
  year = {2021},
  publisher = {{IOP} Publishing},
  volume = {23},
  number = {2},
  pages = {023021},
  author = {Kosuke Mitarai and Keisuke Fujii},
  title = {Constructing a virtual two-qubit gate by sampling single-qubit operations},
  journal = {New Journal of Physics},
  eprint={1909.07534},
  archivePrefix={arXiv}
}

@article{vidal1999robustness,
  title = {Robustness of entanglement},
  author = {Vidal, Guifr\'e and Tarrach, Rolf},
  journal = {Phys. Rev. A},
  volume = {59},
  number = {1},
  pages = {141--155},
  numpages = {0},
  year = {1999},
  publisher = {American Physical Society},
  doi = {10.1103/PhysRevA.59.141},
  url = {https://link.aps.org/doi/10.1103/PhysRevA.59.141},
  eprint={quant-ph/9806094},
  archivePrefix={arXiv}
}

@article{Horodecki2003entanglementbreaking,
  doi = {10.1142/s0129055x03001709},
  url = {https://doi.org/10.1142/s0129055x03001709},
  year = {2003},
  month = aug,
  publisher = {World Scientific Pub Co Pte Lt},
  volume = {15},
  number = {06},
  pages = {629--641},
  author = {Michael Horodecki and Peter W. Shor and Mary Beth Ruskai},
  title = {Entanglement Breaking Channels},
  journal = {Reviews in Mathematical Physics},
  eprint={quant-ph/0302031},
  archivePrefix={arXiv}
}

@article{Yuan2021universalmemories,
  doi = {10.1038/s41534-021-00444-9},
  url = {https://doi.org/10.1038/s41534-021-00444-9},
  year = {2021},
  publisher = {Springer Science and Business Media {LLC}},
  volume = {7},
  number = {1},
  author = {Xiao Yuan and Yunchao Liu and Qi Zhao and Bartosz Regula and Jayne Thompson and Mile Gu},
  title = {Universal and operational benchmarking of quantum memories},
  journal = {npj Quantum Information},
  eprint={1907.02521},
  archivePrefix={arXiv}
}

@article{Lowe2023fast,
  doi = {10.22331/q-2023-03-02-934},
  url = {https://doi.org/10.22331/q-2023-03-02-934},
  year = {2023},
  publisher = {Verein zur Forderung des Open Access Publizierens in den Quantenwissenschaften},
  volume = {7},
  pages = {934},
  author = {Angus Lowe and Matija Medvidovi{\'{c}} and Anthony Hayes and Lee J. O'Riordan and Thomas R. Bromley and Juan Miguel Arrazola and Nathan Killoran},
  title = {Fast quantum circuit cutting with randomized measurements},
  journal = {Quantum},
  eprint={2207.14734},
  archivePrefix={arxiv}
}

@article{bravyi2016trading,
  title = {Trading Classical and Quantum Computational Resources},
  author = {Bravyi, Sergey and Smith, Graeme and Smolin, John A.},
  journal = {Phys. Rev. X},
  volume = {6},
  number = {2},
  pages = {021043},
  numpages = {14},
  year = {2016},
  month = {Jun},
  publisher = {American Physical Society},
  doi = {10.1103/PhysRevX.6.021043},
  url = {https://link.aps.org/doi/10.1103/PhysRevX.6.021043},
  eprint={1506.01396},
  archivePrefix={arXiv}
}

@article{peng2020simulating,
  title = {Simulating Large Quantum Circuits on a Small Quantum Computer},
  author = {Peng, Tianyi and Harrow, Aram W. and Ozols, Maris and Wu, Xiaodi},
  journal = {Phys. Rev. Lett.},
  volume = {125},
  issue = {15},
  pages = {150504},
  numpages = {6},
  year = {2020},
  month = {Oct},
  publisher = {American Physical Society},
  doi = {10.1103/PhysRevLett.125.150504},
  url = {https://link.aps.org/doi/10.1103/PhysRevLett.125.150504},
  eprint={1904.00102},
  archivePrefix={arXiv}
}

@article{piveteau2023circuit,

  author={Piveteau, Christophe and Sutter, David},

  journal={IEEE Transactions on Information Theory}, 

  title={Circuit knitting with classical communication}, 

  year={2023},

  doi={10.1109/TIT.2023.3310797},
  eprint={2205.00016},
  archivePrefix={arXiv}
}

@article{Mitarai2021overhead,
  doi = {10.22331/q-2021-01-28-388},
  url = {https://doi.org/10.22331/q-2021-01-28-388},
  year = {2021},
  publisher = {Verein zur Forderung des Open Access Publizierens in den Quantenwissenschaften},
  volume = {5},
  pages = {388},
  author = {Kosuke Mitarai and Keisuke Fujii},
  title = {Overhead for simulating a non-local channel with local channels by quasiprobability sampling},
  journal = {Quantum},
  eprint={2006.11174},
  archivePrefix={arXiv}
}

@article{Gour2021entanglementbipartite,
  doi = {10.1103/physreva.103.062422},
  url = {https://doi.org/10.1103/physreva.103.062422},
  year = {2021},
  publisher = {American Physical Society ({APS})},
  volume = {103},
  number = {6},
  author = {Gilad Gour and Carlo Maria Scandolo},
  title = {Entanglement of a bipartite channel},
  journal = {Physical Review A},
  eprint={1907.02552},
  archivePrefix={arXiv}
}

@article{harrow2003robustness,
  title = {Robustness of quantum gates in the presence of noise},
  author = {Harrow, Aram W. and Nielsen, Michael A.},
  journal = {Phys. Rev. A},
  volume = {68},
  number = {1},
  pages = {012308},
  numpages = {13},
  year = {2003},
  publisher = {American Physical Society},
  doi = {10.1103/PhysRevA.68.012308},
  url = {https://link.aps.org/doi/10.1103/PhysRevA.68.012308},
  eprint={quant-ph/0301108},
  archivePrefix={arXiv}
}

@article{nielsen2003dynamics,
  title = {Quantum dynamics as a physical resource},
  author = {Nielsen, Michael A. and Dawson, Christopher M. and Dodd, Jennifer L. and Gilchrist, Alexei and Mortimer, Duncan and Osborne, Tobias J. and Bremner, Michael J. and Harrow, Aram W. and Hines, Andrew},
  journal = {Phys. Rev. A},
  volume = {67},
  number = {5},
  pages = {052301},
  numpages = {19},
  year = {2003},
  month = {May},
  publisher = {American Physical Society},
  doi = {10.1103/PhysRevA.67.052301},
  url = {https://link.aps.org/doi/10.1103/PhysRevA.67.052301},
  eprint={quant-ph/0208077},
  archivePrefix={arXiv}
}
\end{document}